\newif\ifsubmission
\newif\ifanon
\newif\ifextendedabstract

\submissiontrue

\documentclass[11pt]{article}

\usepackage{headers}
\title{ 
Diffusion Codes: Self-Correction from Small(er)-Set Expansion with Tunable Non-locality
}
\date{}

\ifanon
    \author{}
\else
    \author[1]{Adithya Sriram}
    \author[1]{Vedika Khemani}
    \affil[1]{{\small 
    Department of Physics, Stanford University, Stanford, CA 94305
    }}
    \author[2]{Benedikt Placke}
    \affil[2]{{\small 
    Rudolf Peierls Centre for Theoretical Physics, University of Oxford, Oxford, UK
    }}
\fi

\begin{document}

\maketitle

\begin{abstract}
Optimal constructions of classical LDPC codes can be obtained by choosing the Tanner graph uniformly at random among biregular graphs. We introduce a class of codes that we call ``diffusion codes'', defined by placing each edge connecting bits and checks on some graph, and acting on that graph with a random SWAP network. By tuning the depth of the SWAP network, we can tune a tradeoff between the amount of randomness --- and hence the optimality of code parameters --- and locality with respect to the underlying graph.
For diffusion codes defined on the cycle graph, if the SWAP network has depth $\sim Tn$ with $T> n^{2\beta}$ for arbitrary $\beta>0$, then we prove that almost surely the Tanner graph is a lossless ``smaller set'' vertex expander for small sets up size $\delta \sim \sqrt T \sim n^{\beta}$, with bounded bit and check degree. At the same time, the geometric size of the largest stabilizer is bounded by $\sqrt T$ in graph distance. 
We argue, based on physical intuition, that this result should hold more generally on arbitrary graphs.
By taking hypergraph products of these classical codes we obtain quantum LDPC codes defined on the torus with smaller-set boundary and co-boundary expansion and the same expansion/locality tradeoffs as for the classical codes.
These codes are self-correcting and admit single-shot decoding, while having the geometric size of the stabilizer growing as an arbitrarily small power law.
Our proof technique establishes mixing of a random SWAP network on small subsystems at times scaling with only the subsystem size, which may be of independent interest.
\end{abstract}

\tableofcontents

\ifextendedabstract
\else
    \section{Introduction}
\fi

\ifextendedabstract
    \noindent
\else
\fi

A major challenge in quantum error correction is the inherent tradeoff between favorable properties of codes constructed in high spatial dimensions and the need for locality in practical implementations on near term hardware. A canonical example is passive quantum memory: it is known to exist for local models in $D \geq 4$, forbidden for $D \leq 2$, and remains elusive in three dimensions \cite{Dennis_2002, BPTbounds2, Bravyi_2013, lin2024proposals3dselfcorrectingquantum}. The contrast is sharper still between local Euclidean lattices and non-Euclidean expander graphs, which are effectively ``infinite-dimensional''.
%
%
Such graphs underlie recent breakthroughs in quantum error correction which have produced ``good'' quantum low density parity check (LDPC) codes \cite{breuckmann_balancedproductcodes,panteleev2022asymptoticallygoodquantumlocally, dinur2022goodquantumldpccodes, leverrier2022quantumtannercodes} with optimal rate and distance scaling. Furthermore, these codes and related constructions are known to admit linear-time decoding algorithms \cite{Fawzi_2018_lineartime, dinur2022goodquantumldpccodes}, self correction \cite{Hong_2025,placke2024topologicalquantumspinglass}, and single-shot error correction \cite{Fawzi_2018, gu2024single_shot}.


What makes all these codes ‘good’ is expansion: expander graphs have large boundary-to-volume ratios and small diameters, properties that translate into constant rate and large distance for LDPC constructions. The catch is that such codes are not geometrically local: their parity checks couple qubits that are far apart in any low-dimensional Euclidean embedding. By contrast, in any fixed spatial dimension with local checks there are rigorous tradeoffs that preclude ‘good’ quantum LDPC codes and severely constrain mechanisms for passive memory \cite{BPTbounds1, BPTbounds2, YOSHIDA20112566}. This has sparked interest in a middle ground: physically realizable architectures that preserve some expander-like connectivity---such as codes with sparse long-range links, or fractal geometries.
Exploring this middle regime is compelling both practically (lower overhead, higher thresholds, and potentially larger energy barriers) and conceptually, as it interpolates between the rigorously understood, effectively infinite-dimensional limit and the physically realized local Euclidean setting. 
This has motivated a program of study towards this effort \cite{Yang_2025, PRXQuantum.6.010306, Berthusen2024partialsyndrome, dai2024localityvsquantumcodes,williamson2024layercodes,baspin2024wirecodes}.

Against this backdrop, we develop a concrete “middle-ground’’ construction that interpolates between fully local Euclidean lattices and fully nonlocal expanders. We introduce a family of  codes which we call \emph{diffusion codes}, with a controllable knob that trades locality for expansion, yielding classical (and, via hypergraph product, quantum) LDPC codes with, what we term, ``smaller set expansion''. In typical expander codes, errors up to a size which is extensive in the total number of bits trigger a syndrome which grows at least linearly with the size of the error. This is known in the literature as "small set expansion." In contrast, our diffusion codes are smaller set expanders, which means only errors up to a size which is \textit{sub-extensive} in system size have a linearly growing (in error weight) lower bound on the syndrome weight. We define these terms in more detail in \cref{app:preliminaries}. We state this main result as an informal theorem below:

\begin{theo} [Main Result, Informal]
    For all $\beta > 0$, there exist classical $[n, \mathcal{O}(n), \mathcal{O}(n^\beta)]$ LDPC codes with $\beta > 0$, and an embedding of these codes onto the line such that
    \begin{enumerate}
        \item the geometric size of every check (the diameter of its support on the line) is concentrated around $\sim  \mathcal{O}(n^\beta)$
        \item the code has the property of ``smaller set'' expansion, i.e. there exist constants $(\gamma, \delta)$ for any error $\mathbf{x}$ such that $|\mathbf{x}| \leq \delta \cdot n^\beta$, the weight of the associated syndrome $\mathbf{s}$ is lower bounded as $\mathbf{s} \geq \gamma\cdot |\mathbf{x}|$.
    \end{enumerate} 
    By taking the hypergraph product of two such codes, one can further obtain a family of quantum LDPC codes with the same locality in two dimensions and the same ``smaller set'' expansion property.
\end{theo}

The construction of these codes arises from running a local shuffling process on the Tanner graph for a fixed amount of time. The process resembles diffusion, whence the name. We rigorously prove that this results in codes with smaller set expansion on the cycle graph, which yields a quasi-1D classical code (quasi-2D quantum code via the hypergraph product). However, based on the physically expected generality of diffusion ideas, we expect the results to actually hold on arbitrary graphs.

Diffusion codes have bits (qubits) which are arranged on a finite dimensional Euclidean lattice. The models are non-local, but the checks of the code involve bits which are all contained within a tunable radius of size $\delta(n) \sim n^{\beta}$, which scales sub-extensively with system size ($\beta < 1$). 
This radius is tuned by the time that the shuffling process is run. The result is that locally within this radius, the interaction graph of the code resembles that of an expander code; globally, the Euclidean structure reemerges. We prove that even for \emph{arbitrarily small} $\beta > 0 $, the code retains small(er) set expansion: subsets of vertices up to size $\delta(n)$ have a boundary of a size proportional to the subset size. 
It follows that classical diffusion codes can have parameters scaling as $[n, \mO (n), \mO (n^\beta)]$. Importantly, the smaller-set expansion property is still sufficient to prove that our codes are self-correcting. We show that the quantum codes obtained as the hypergraph product of two classical diffusion codes inherit the same smaller-set expansion and self-correction properties. We conclude by showcasing some numerical experiments that illustrate the properties that we prove. 

We prove the above results by relating expansion in diffusion codes to properties of the joint distribution of interparticle spacings in a simple exclusion process induced by the aforementioned shuffling process. The key technical challenge is to establish properties of this distribution for arbitrary initial states and at ``mesoscopic'' times $T$, which are longer than the subsystem size $\delta$ but much shorter than the global mixing time ($\delta \ll T \ll t_{\rm mix}$). We expect our proof technique for achieving this to be of independent interest. 

Finally, we note that from a physical perspective, good classical expander LDPC codes, when viewed as spin Hamiltonians, are known to be in a \textit{spin glass phase} \cite{mezard_montanari}. This topic was revisited recently and shown to arise directly from the underlying expansion of the code \cite{placke2025expansioncreatesspinglassorder}. Furthermore, as expansion is also a property of the good qLDPC codes, Hamiltonians of these codes also exhibit an inherently quantum version of the spin glass phase, termed a topological quantum spin glass \cite{placke2024topologicalquantumspinglass}. 
In both cases, code expansion, in addition to self-correction (that is, stability of the ground state), creates a complex free energy landscape with many local minima.
Our work here suggests that this physics of the expander models may be retained also in models with a well-defined notion of spatial locality. Indeed, we show that the central ingredient which gives rise to the glass physics of ``good'' codes, expansion, is also realized in diffusion codes. Because of this, we expect them to also furnish a glass phase, both in the classical and quantum setting, and our numerical results are consistent with this expectation. The problem of understanding finite dimensional spin glasses remains one of the paradigmatic open problems \cite{moorestein, PhysRevB.35.6841, Altieri_2024, spinglassbook} in statistical physics, and our work represents a contribution to this effort as well.




In the remainder of this section, we informally explain the general construction of diffusion codes, the physical intuition behind them, and summarize our rigorous results.



\subsection{Construction of Diffusion Codes\label{sec:summary_diffusion_codes}}

\begin{figure}
    \centering
    \includegraphics[width=0.85\linewidth]{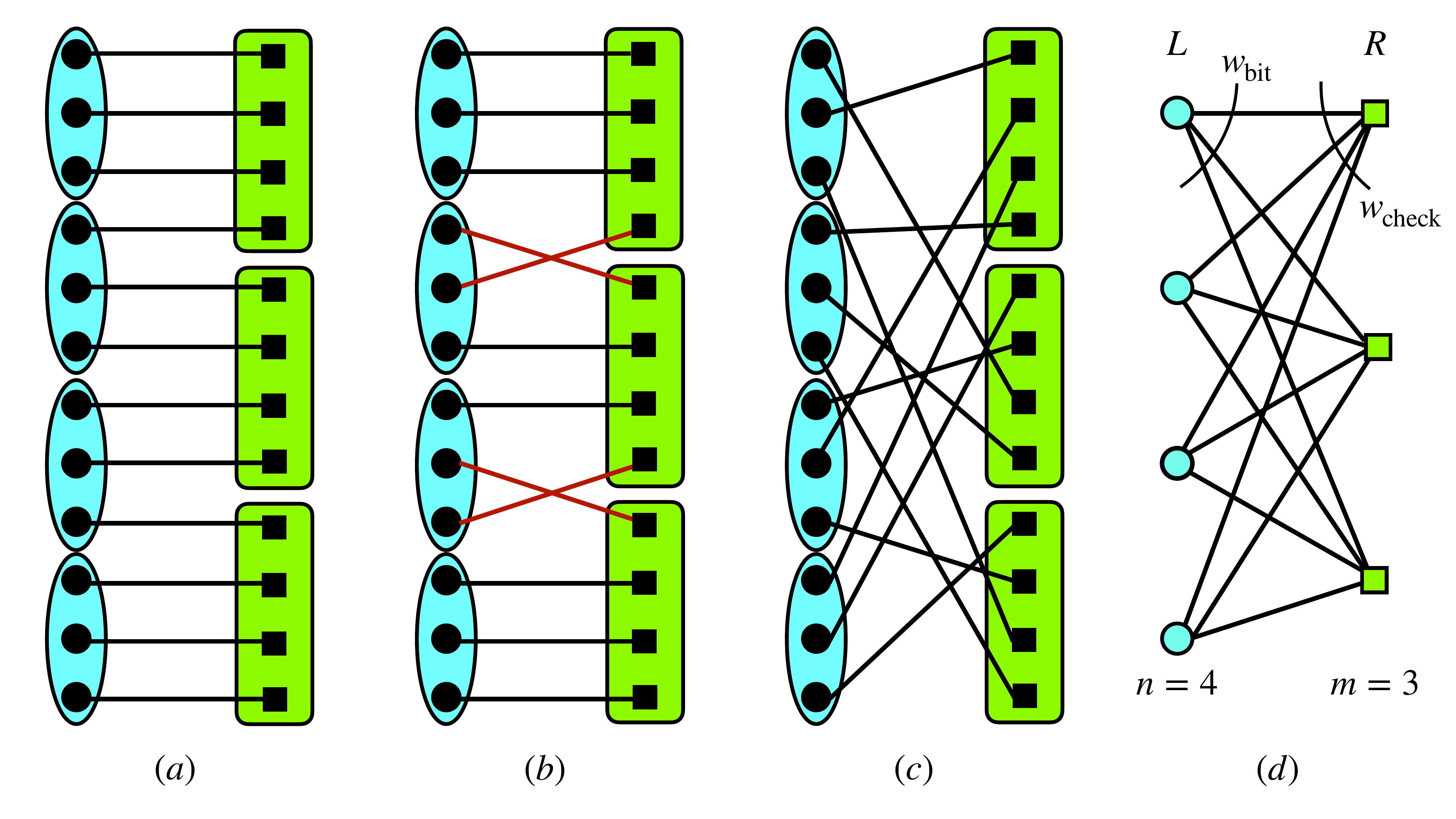}
    \caption{Sketch of the diffusion code construction on a line graph with $n=4$ bits, $m=3$ checks, $\wbit = 3$ and $\wcheck=4$. (a) We start from a perfect matching of $N=n\wbit=m\wcheck$ nodes, grouped in two ways corresponding to bits and checks. (b) We act of with a sequence of random pairwise SWAPS on the nodes, for a total of $T$ time steps (corresponding to $NT$ SWAPs in total). (c) After time $T$ we have created some permutation of the $N$ nodes. (d) We then collapse the node groups into the left and right vertices of a bipartite graph, which we take to be the Tanner graph of code. 
    }
    \label{fig:summary_diffusion_codes}
\end{figure}

The oldest construction of good classical LDPC codes is due to Gallager \cite{gallagerLDPC}. These codes are most conveniently represented by their Tanner graph. For a Gallager code, this is a random bi-regular graph between $n$ left-vertices (bits) and $m$ right-vertices (checks) with degree $\wbit$ and $\wcheck$, respectively.
How does one choose the graph (and hence the code) in practice? A common prescription is the \emph{configuration model} \cite{Richardson_Urbanke_2008,gallagerLDPC,sipserspielman}: first choose a random matching between $N = n\wbit$ left nodes and $ N = m\wcheck$ right nodes, and then group consecutive sets of $\wbit$ left nodes and $\wcheck$ right nodes to form the left and right vertices of the graph, respectively, see also \cref{fig:summary_diffusion_codes} (c-d). 

Such a random graph, when interpreted as the Tanner graph (see \cref{def:tanner graph}), defines a good classical code with constant rate $k\sim n$ and optimal distance scaling $d\sim n$. Sipser and Spielman \cite{sipserspielman} further showed that such codes are \emph{expander codes}, with energy barriers that scale linearly with $n$ between codewords, and which admit an efficient (linear time) decoder. In the context of quantum error correction, hypergraph products of such classical expander codes were introduced by Leverrier, Tillich, and Zemor \cite{Leverrier_2015}. While these ``quantum expander codes'' do not have optimal parameters ($k\sim n$ but $d\sim \sqrt{n}$), they are known to admit linear time decoders~\cite{quantumexpandercodes,Fawzi_2018} and single-shot error correction\cite{Fawzi_2018,Gu_2024} for \emph{random errors}, and to act as passive quantum memories \cite{Hong_2025, placke2024topologicalquantumspinglass}.

The favorable properties of Gallager codes and their hypergraph products are all inherently linked to the fact that their Tanner graph $\mB = (L, R, E)$ is a lossless small-set bipartite vertex expander. Given any subset of bits $S \subset L$, with $\abs{S}\leq \delta n$ for some $\delta > 0$, the size of the neighbor set $\Gamma(S) \subset R$ is proportional to the size of the set itself $\abs{\Gamma(S)} \geq \wbit(1-\varepsilon)\abs{S}$, where the constant $\varepsilon$ can be made arbitrarily small.

Due to the underlying expansion, Gallager codes cannot be made geometrically local in any finite Euclidean dimension. The goal of diffusion codes is to provide a (stochastically) uniform construction which preserves some degree of locality with respect to a specified underlying geometry, while at the same time preserving the favorable properties of Gallager codes as much as possible. 

\begin{figure}
    \centering
    \includegraphics[width=0.75\linewidth]{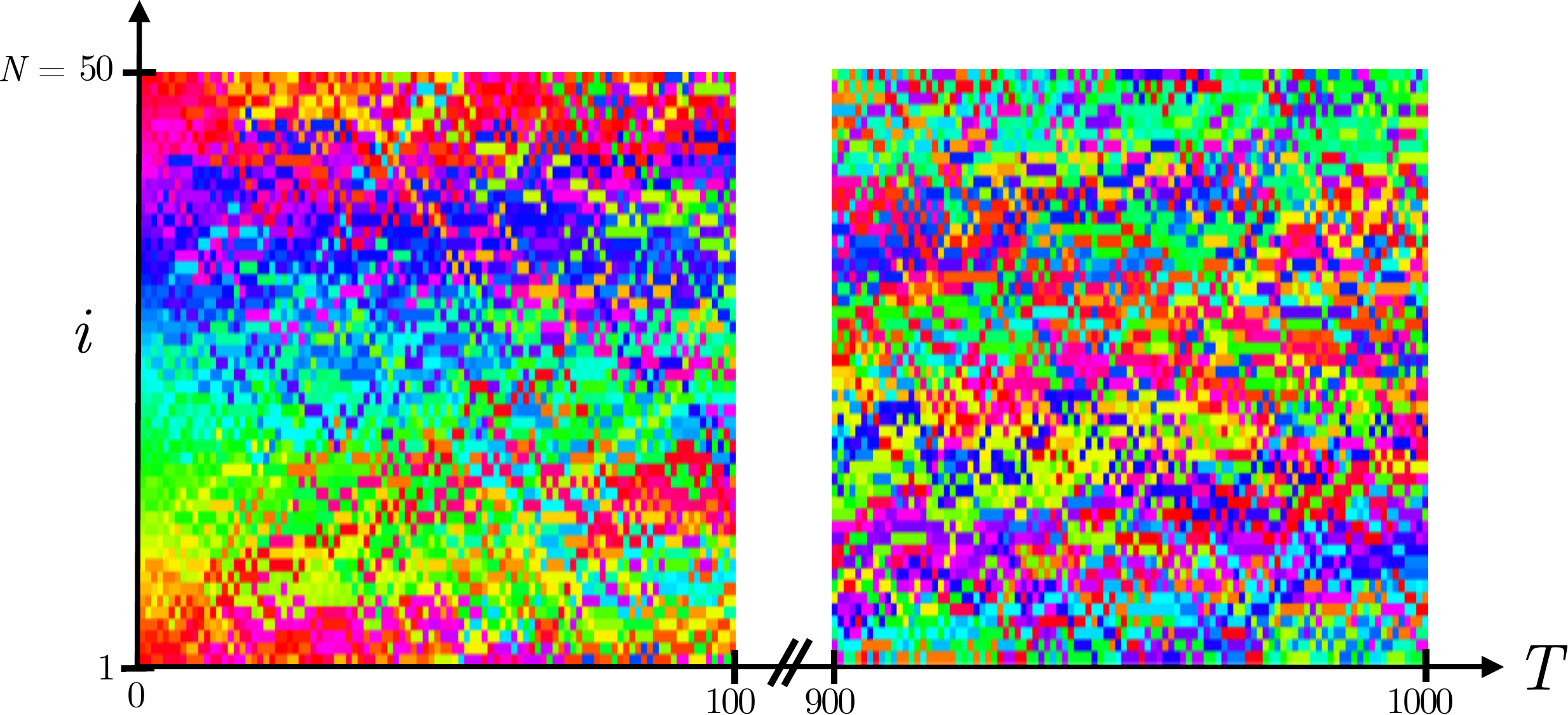}
    \caption{Illustration of mixing in the random SWAP process. We indicate each node by a different color, arranged in a regular gradient at the start.  Every time step corresponds to $n$ randomly chosen SWAPs. For times $T\ll N^2$ the process is not fully mixed and the resulting permutation retains some degree of locality. Nevertheless, it is \emph{locally} well mixed.}
    \label{fig:diffusion_process}
\end{figure}

The idea of our construction is inspired by the configuration model. We generate the Tanner graph randomly by drawing a permutation of $N$ elements. However, we want to draw this permutation not uniformly random among all permutations, but in such a way that it looks ``locally random'' while also preserving some degree of locality.
To this end we allude to intuition from physics. A completely uniformly random permutation of $N$ variables can be generated by running a \emph{local} random process for a time $T \gg N$. The idea of our construction will be to limit the amount of time $T$ that this process is run, to preserve some degree of locality.
In particular, it is known that a random permutation can be generated by running a local SWAP process (also called the interchange process or random adjacent transpositions) for $\sim N^2$ steps (in our convention, we perform on average $N$ SWAPs per step) \cite{markovchains}. This is the process that we will generalize. We illustrate the time evolution of a random SWAP circuit on $n=50$ nodes in \cref{fig:diffusion_process}, with each node indicated by a different color. The initial state is a clean color gradient, and for short times similar colors stay close to each other but their order is mixed up locally.
In the local SWAP process, any single node undergoes a diffusion process over time and its position delocalizes over a range scaling as $\sqrt T$. This motivates the name of our construction---diffusion codes.


Concretely but informally (see \cref{def:diffusioncodes} for a more formal definition), our construction is as follows. Given a graph $\mG$ of $N$ vertices we do the following
\begin{enumerate}
    \item Create a random permutation of the $N$ vertices by acting on them with a random circuit of nearest-neighbor (with respect to $\mG$) SWAP gates of depth $NT$.
    \item Match each vertex with its initial position to create a bipartite graph of $N$ left and $N$ right vertices. 
    \item Given two partitions of the vertices into groups of size $\wbit$ and $\wcheck$, respectively, collapse these groups into left and right-vertices, to produce a $(\wbit, \wcheck)$ bi-regular graph.
    \item Define a code by interpreting this graph as a Tanner graph of $n = N / \wbit$ bits and $m = N / \wcheck$ checks.
\end{enumerate}
The construction is illustrated in \cref{fig:summary_diffusion_codes}. We can obtain quantum LDPC codes by taking hypergraph products of diffusion codes. 

If we follow the time trace of a single vertex---corresponding to a single edge of the Tanner graph---over time, then this vertex is exactly undergoing a random walk on $G$. 
However, the random walks carried out by different edges are not independent, and the expansion of the Tanner graph is a property of their joint distribution.
While this is the main challenge in proving results on diffusion codes in the general setting, we nevertheless expect---based on the physical intuition---that general diffusion codes will locally behave like Gallager codes. 
In particular we expect them to have the lossless \textit{smaller set} expansion property: 
We say that the Tanner graph $\mathcal B = (L, R, E)$ is a lossless $(\delta(n), \gamma)$ smaller-set expander if for all $S\subset L$ with $\abs{S}\leq \delta(n)$, we have $\abs{\Gamma(S)}\geq \gamma \abs{S}$ with $\gamma = \wbit (1-\varepsilon)$ as before. The smaller-set cutoff $\delta(n)$ here is a sublinear function of $n$. For diffusion codes, we expect it to scale as $\delta \sim \sqrt T$, and hence to be freely tunable by the choice of the diffusion time $T$ as a function of $n$.
In the limit of $T\sim n^2$, we recover the configuration model and hence Gallager codes, independent of the choice of the underlying graph $G$.

Smaller-set expansion of the Tanner graph has many direct consequences for the properties of the resulting code. For $\gamma > \wbit /2$ (or equivalently, $\varepsilon < 1/2$), it implies code expansion up to error size $\delta(n)$ in the sense of Sipser and Spielman (\cref{def:expander codes}). This in turn implies that as long as $\delta(n)$ scales as an arbitrarily small power law $\delta\sim n^\beta$ for some $\beta > 0$, the codes are self correcting (see \cref{sec:bottleneck}). 
For $\gamma > 3\wbit / 4$, \cite{quantumexpandercodes} smaller-set expansion implies code expansion up to error size $\delta(n)$ in the hypergraph product of the code with itself. Again, as long as $\delta\sim n^\beta$ for some $\beta < 1$, this implies quantum self correction.

\subsection{Summary of Rigorous Results} \label{sec:summary_results}

We now summarize our rigorous results, which make the above intuition rigorous for the case of diffusion codes based on the cycle graph $\mathcal C_N$ and their hypergraph products.

For this summary to be self-contained, we reiterate that we call a bipartite graph $\mB = (L,R,E)$ with left degree $c$ and right degree $d$ a $(\delta, \gamma)$ small-set expander if $\forall S \subset L \ : \ |S| \leq \delta$, the neighbor set of $S$, $\Gamma(S) \subset R$ has the property that $|\Gamma (S) | \geq \gamma |S|$ (see \cref{def:vertexexpander}). If $\delta(n) = o(n)$, then we say the graph is a \textit{smaller set} expander. We furthermore say that a family of graphs is a family of lossless expanders if $\gamma = c(1-\varepsilon)$ where $\varepsilon \to 0$ as $c \to \infty$. When $\mB$ is the Tanner graph of a code, $c$ defines $\wbit$ and $d$ defines $\wcheck$. 


Our main results are the following:
\begin{enumerate}
    \item For the ``diffusion code'' construction defined above, we show that if they are constructed from applying the SWAP network on the cycle graph $\mathcal C_N$ and if $T \sim N^\alpha$, then almost surely the resulting Tanner graph is a smaller-set expander for any $\delta\sim n^\beta$ for any $\beta > \alpha /2$ (\cref{theorem:diffusioncodecycle}). Furthermore, every check is almost surely no larger than $ \sim T^{1/2}$.

    \item To show the above result, we relate the question of expansion in the Tanner graph, to a property of the distribution of interparticle distances in a related simple exclusion process (SEP) on $\mathcal C_N$. The central technical challenge is that the time $T$ at which the distribution is inspected is much smaller than the mixing time of this SEP. In particular, we are interested in a property of the joint distribution of \emph{many} particles $k\sim \sqrt T$, when starting from arbitrary initial states. 
    We solve this challenge by defining a Markov chain that acts directly on the interparticle distances, and establish stochastic monotonicity between this chain and one corresponding to the SEP on a smaller cycle graph $\mathcal C_{N'}$. Choosing $N'$ such that the mixing time of the SEP on $\mathcal C_{N'}$ is smaller than $T$, we can then use properties of the steady-state distribution of this SEP, together with monotonicity, to obtain a bound on the interparticle distances also in the SEP on $\mathcal C_{N}$.
    We expect this proof technique to be of independent interest, for example in the study of local scrambling and pseudo-randomness in  unitary circuits.

    \item Even the ``smaller set expansion'' property established for diffusion codes suffices to guarantee a number of desirable properties. The classical codes are thermally stable \cite{spinglassbook,placke2025expansioncreatesspinglassorder} under Glauber dynamics (\cref{theorem:self correction from confinement}). Furthermore, the proofs for linear (co-)boundary confinement of the hypergraph product in \cite{Leverrier_2015} extend readily to diffusion codes. 
    As a result, the hypergraph products of diffusion codes are qLDPC codes which are self correcting. 
    We also expect related consequences of expansion to generalize, such as linear-time single shot decoding against random noise \cite{Fawzi_2018_lineartime,Fawzi_2018}. 
    Hypergraph products of diffusion codes on the cycle graph admit a natural embedding on the torus, with the maximum geometric non-locality of any check almost surely given by $\delta(n)$ (\cref{theorem: sublinear expander hypergraph product}).
\end{enumerate}

As a primer to the diffusion code construction, we also more generally prove the existence of smaller set expanders through a much simpler random construction in \cref{sec:proof_existence}. However, the problem in this construction is that it is not symmetric between bits and checks. This means that (i) the right degree (corresponding to $\wcheck$ of the code) is not constrained by construction and (ii) it is not clear that the resulting graph is both a left- and right- small set expander.

\subsection{Summary of Numerical Experiments}


We perform a number of numerical experiments that illustrate the properties of diffusion codes. 

We determine the threshold of classical diffusion codes under the local flip decoder of \cite{sipserspielman} and under belief propagation. For a particular choice of code parameters (i.e. $\wbit = 9, \wcheck = 11, T \sim n^{1}$), we find the threshold to be $0.017 \leq p_c^{\rm flip} \leq 0.019$ and $0.11 \leq p_c^{\rm BP} \leq 0.13$ for the local decoder and the belief propagation decoder respectively.
We also study the self-correction capability of these codes by conducting a "memory time" experiment. Self correction would indicate that when the code is subjected to thermal noise, the time at which a decoder can no longer recover the state diverges in system size. In memory time experiments for classical diffusion codes, we find that the memory time as measured using the local flip decoder state diverges as a stretched exponential in system size. 
We also show the result of a ``heating" experiment, which demonstrates that the codes remain out of equilibrium for long time at low temperatures, and a ``cooling" experiment which shows that the codes fall out of equilibrium at some low temperature when annealed from a high temperature. In the cooling experiments, we find that under annealing, the system fails to reach the ground state which is consistent with spin glass physics. 

Finally, we also perform such ``heating-cooling" experiments on quantum codes constructed as hypergraph products of diffusion codes, to provide a demonstration of their self-correction capability. We furthermore find that, similar to the classical codes, under annealing in the cooling experiment, the system fails to find the ground state. This is consistent with the phenomenology of the recently proposed topological quantum spin glass phase \cite{placke2024topologicalquantumspinglass}. 

\subsection{Related Work}

Here, we give a brief overview and comparison of related work. 

\textbf{Spatial Coupling of LDPC Codes.} Spatially coupled LDPC codes are a family of codes that are constructed by linking together copies of a single \textit{protograph code} \cite{sc-LDPC1, sc-LDPC2}. The protograph code is usually a good LDPC code of some size and copies of it are linked together by changing a few edges so as to connect nodes between them. This yields a parity check matrix which is banded. Innovative methods of coupling the protographs together can lead to improved performance under belief propagation \cite{sc-LDPC2, sc-LDPC3}. These methods, originally developed for classical codes, have recently been extended to quantum codes \cite{Yang_2025}.

Diffusion codes in 1D constructed from the cycle graph are in some ways similar to this idea. The size of the protograph is analagous to the scale of expansion. In fact, consider a diffusion code of size $n$ constructed from a diffusion time of $T \sim n^{2\beta}$. This code has code parameters $[n, \mathcal{O}(n), \mathcal{O}(n^\beta)]$. These same code parameters could be achieved by placing $n^{1-\beta}$ different "good" codes of size $n^\beta$ along a line, even without coupling them. What sets diffusion codes apart from these is the simplicity and flexibility of their construction (they naturally can be defined on any graph geometry). Diffusion codes are also by construction stochastically homogeneous with respect to the underlying graph, making them more natural from a physics perspective, and also potentially easier to implement.

\textbf{Interpolation from Topological Codes.} Another approach to achieve tunable non-locality proposed in recent years has been to add non-locality to topological codes. This may be done either by modifying some number of checks to make them long range \cite{Hong_2024}, or to insert additional checks and bits to ``augment'' the topological code \cite{Roffe_2020}. Yet another approach is to concatenate an appropriate qLDPC code with surface codes, in which case the quasi-locality may be tuned by varying the input qLDPC code \cite{dai2024localityvsquantumcodes}. 
This last construction has optimal code parameter tradeoffs, in the sense that itsaturates generalized BPT bounds, which were presented in the same work. These generalized bounds answer the question of how many long range interactions in 2D are required in order to circumvent the BPT bound. Dai and Li proved that any $[n,k,d]$ 2D quantum stabilizer code with $kd^2 \geq \mathcal{O}(n)$ must have at least $c_0 \cdot \text{max}(k,d)$ interactions of length at least $c_0 \cdot \text{max} \left( \frac{d}{\sqrt{n}}, \left( \frac{kd^2}{n}\right)^{1/4}\right)$. 
By concatenating a good qLDPC code with a surface code and implementing the result on a ``stacked architecture'' \cite{dai2024localityvsquantumcodes} they are further able to interpolate the trade-off between the number of nonlocal checks, and their size.

These approaches are interesting, and allow for much flexibility. Our construction is quite different in nature, and in contrast to all approaches summarized above naturally homogeneous in space. In terms of parameters, in our construction we have many checks which will typically be of size $L^\beta$, where $\beta$ can be made arbitrarily small (at the cost of also decreasing the distance by the same amount). We emphasize that while out construction is sub-optimal from a pure code parameter perspective, we proof several properties of diffusion code \emph{beyond} just parameter scaling, such as expansion, and as a consequence self-correction and single-shot decoding.

\textbf{Embedding from higher dimensions.} 
Recently, there have been a few examples of codes saturating the BPT bounds (\cite{BPTbounds1}) in finite dimensions. These codes are all constructed via algebraic embeddings of higher dimensional optimal codes into lower dimensions while sacrificing as little favorable properties of the input codes as possible \cite{williamson2024layercodes, portnoy2023localquantumcodessubdivided, baspin2024wirecodes}. 
More generally, from theory of metric embeddings, there exist methods to embed graphs into Euclidean dimensions (e.g. Johnson-Lindenstrauss lemma) at the cost of distortion and some nonlocality \cite{woottersrandomalgos}.
However, in general at least some favorable properties of the input codes will be lost. For example, it has recently been shown that layer codes are not self-correcting, even if the input codes are \cite{baspin2025freeenergybarriereyringpolanyi}. 
Baspin \cite{baspin2023combinatorialstructureslinearcodes} further showed that the connectivity graph of \emph{any} $[n, \Theta(n), d]$ quantum code must contain expanders of size $\Omega(d)$ as subgraphs, and the union of these subgraphs must cover a finite fraction of the full connectivity graph. This is similar (although not identical) to the Tanner graph of the code being an $\Omega(d)$ ``smaller set'' expander. In particular, the results do not allow for the recovery of \emph{lossless bipartite} expanding subgraphs, which is necessary to guarantee the expansion properties of the hypergraph product.


\textbf{Long range models.} When viewed as Hamiltonians, the diffusion codes we construct, especially from the cycle graph, can be viewed as long-range models, where the scale of any interaction has spread over a power law in system size. This construction takes some inspiration from a vast body of literature in physics concerning long-range models. Indeed a famous example is that of the 1D Ising model. Although the short-range version has no phase transition, if one instead considers a 1D Ising model where spins interact with a strength that falls off as a power law in distance, then, depending on the strength of the power law, the resulting model \textit{does} have a phase transition \cite{dyson1disingpowerlaw}. Another example is a one-dimensional embedding of the Sherrington-Kirkpatrick (SK) model of spin glasses presented in \cite{Leuzzi_2008}. The authors embed the SK model onto a 1D lattice by considering an interaction matrix $J_{ij}$, where each entry has a probability of being non-zero which decays as a power law in $|i-j|$. Through extensive numerical experiments, they present evidence of replica symmetry breaking and spin glass phenomenology in this model.
\section{Preliminaries}\label{app:preliminaries}

In this section, we provide essential background on understanding the proofs behind the constructions presented in the main text. For greater detail, we refer the reader to Ref. \cite{essentialcodingtheory}. We begin with some basic graph theory definitions. 

\subsection{Basic Graph Theory}

\begin{defn} [Graphs] 
    A graph $\mathcal{G} = (V,E)$ consists of a set $V$ and a set $E \subseteq V \times V$.
\end{defn}

If a graph has edges which may be bipartitioned into vertex sets where edges connect vertices between sets but not within sets, then we refer to such an object as a bipartite graph.

\begin{defn} [Bipartite graphs]
    A bipartite graph $\mathcal{B} = (L,R,E)$ consists of a set $L$, a set $R$ and a set $E \subseteq L \times R$. 
\end{defn}

We refer to $L$ as "left vertices" and $R$ as right vertices. 

\begin{defn} [Adjacency matrix] 
    Any graph $\mathcal{G}$ may be defined by its \textit{adjacency matrix}, $\mathbf{A}_{\mathcal{G}} \in \mathbb{F}_2^{|V| \times |V|}$, with the matrix elements of $\mathbf{A}_{\mathcal{G}}$ given by
    \begin{align} 
        (\mathbf{A}_{\mathcal{G}})_{ij} = \begin{cases}
            1\ \text{if } (v_i, v_j) \in E \\
            0\ \text{else}
        \end{cases}.
    \end{align}
\end{defn}

For a bipartite graph, the adjacency matrix contains a block structure
\begin{align}
    \mathbf{A}_{\mathcal{B}} = 
\left( \begin{array}{c|c}
0 & \mathbf{H} \\
\hline
\mathbf{H}^{\rm T} & 0
\end{array} \right).
\end{align}

The matrix $\mathbf{H}$ is referred to as the biadjacency matrix. 

\begin{defn} [Biadjacency matrix] 
    The \textit{biadjacency matrix} of a bipartite graph $\mathcal{B}$ is a matrix $\mathbf{H} \in \mathbb{F}_2^{|L|\times|R|}$ given by 
    \begin{align}
        (\mathbf{H}_{\mathcal{B}})_{\ell, r} = \begin{cases}
            1\ \text{if } (\ell, r) \in E \\
            0\ \text{else}
        \end{cases},
    \end{align}
    where $\ell \in L$ and $r \in R$.
\end{defn}

\begin{defn} [Vertex Regularity/Boundedness]
    Given a graph $\mathcal{G}$, we say if it is $c$-regular (bounded) if every vertex in the graph is only part of $c$ ($\leq c$) edges. For a $c$-regular (bounded) graph, $\sum_i (\mathbf{A}_{\mathcal{G}})_{ij} = c (\leq c)$.
\end{defn}

\begin{defn} [Left/Right Regularity/Boundedness] \label{def: left right boundedness}
    A bipartite graph $\mathcal{B}$ is $(c,d)$-regular (bounded) if every left vertex is only part of $c$ ($\leq c$) edges and every right vertex is only part of $d$ ($\leq d$) edges. For a $(c,d)$-regular (bounded) graph, $\sum_\ell (\mathbf{H})_{\ell, r} = c (\leq c)$ and $\sum_r (\mathbf{H})_{\ell, r} = d (\leq d)$.
\end{defn}

\begin{defn} [Neighbor sets]
    Given a graph $\mathcal{G}$ and a vertex $u \in V$, the \textit{neighbor set} of that vertex is the set of all vertices connected to that vertex. That is, the neighbor set $\Gamma(u)$ is
    \begin{align}
        \Gamma (u) = \{ v : (u,v) \in E\}.
    \end{align}
\end{defn}

For a bipartite graph, $\forall \ell \in L$ ($r \in R$), $\Gamma(\ell) \in R$ ($\Gamma (r) \in L$).

\begin{defn}[Unique neighbor sets]\label{def:unique_neighbor_set}
    Given a graph $\mathcal{G}$ and a set of vertices $S \subseteq L$, the unique neighbor set $\Gamma_u(S)$ is the set of vertices which connect \textit{only once} into $S$. 
\end{defn}

We now define what it means for a graph to be a vertex expander.

\begin{defn} [Vertex expansion] \label{def:vertexexpander}
    A graph $\mathcal{G} = (V,E)$ is said to be a $(\delta, \gamma)$ vertex expander if for every subset of vertices $S \subseteq V$ such that $|S| \leq \delta $, we have $|\Gamma(S)| \geq \gamma |S|$.
\end{defn}

In the case of bipartite graphs, we may similarly define notions of left (right) expansion.

\begin{defn} [Left/right vertex expansion] \label{def:leftright expansion}
    We say a bipartite graph $\mathcal{B} = (L,R,E)$ is a left (right) $(\delta, \gamma)$ expander if for every subset of left (right) vertices $S \subseteq L(R)$ with $|S| \leq \delta$, we have $|\Gamma (S)| \geq \gamma |S|$.
\end{defn}

We make several remarks about the above definition. First, for a subset of left vertices $S \subseteq L$, $|\Gamma(S)| \leq |R|$. This implies that $\gamma \leq |R| /  \delta$. Next, if $\mathcal{B}$ is $c$-left regular/bounded, then $\gamma \leq c$. 

In general, having $\gamma  = 1$ is not surprising. For example, the bipartite cycle graph is a left/right vertex expander with $\gamma= 1$. What is nontrivial is attaining $\gamma > c /2$. For families of graphs in which there is a notion of taking a "thermodynamic limit," i.e. a limit of large $|L|, |R|$ with $|R| \leq |L|$, if $\gamma \to c$, we say that the family of graphs are \textit{lossless expanders}. 

\begin{defn} [Lossless expansion] \label{def:lossless expansion}
    Consider a family of bipartite graphs $\{\mB_i\}$, $i\in \mathbb{N}$ with $n_i$ left vertices and $m_i$ right vertices. We say this is a family of  lossless left (right) vertex expanders if for every $i$, $\mB_i$ is a $(\delta(n_i), \gamma)$ expander with $\gamma  = c(1-\varepsilon)$, and as $\varepsilon \to 0$, $c \to \infty$.
\end{defn}

The magnitude of $\gamma$ (i.e. how small $\varepsilon$ is) is a measure of the expansion strength. Strong enough vertex expansion implies a more stringent form of expansion known as \textit{unique neighbor expansion}.

\begin{defn} [Unique neighbor expansion] \label{def:unique neigbhor expansion}
    Let $\mB = (L,R,E)$ be a a bipartite graph with $|L| = n$ and $|R| = m$. We call $\mathcal B$ a left (right) $(\delta, \gamma)$ unique-neighbor expander if for every $S \subseteq L (R)$ with $|S| \leq \delta$, we have $|\Gamma_u(S)| \geq \gamma |S|$, where $\Gamma_u(S)$ is the unique-neighbor set of $S$ (\cref{def:unique_neighbor_set}).
\end{defn}

\begin{lem}  [Unique neighbor expansion from lossless vertex expansion ] \label{lemma:unique neighbor exp}
    Let $\mathcal{B}_i = (L,R,E)$ belong to a family of left $c$-regular $(\delta, \gamma)$ lossless expanders, with $\gamma = c(1-\varepsilon)$. Then $\forall \varepsilon \in (0,1/2)$, $\forall S \subseteq L$ with $|S| \leq \delta$, we have $|\Gamma_u(S)| \geq c(1-2\varepsilon) |S|$.
\end{lem}

\begin{proof}
    Consider $S \subseteq L$, $|S| \leq \delta$. By expansion, $|\Gamma(S)| \geq c(1-\varepsilon)|S|$. By $c$-left regularity, the number of edges emanating from $S$ is $c|S|$. For every $v \in \Gamma(S)$, designate one edge leaving $v$ as "special." Note that the number of non-special edges is at most $\varepsilon c |S|$. If a $v \in \Gamma(S)$ contains only a special edge then it is a unique neighbor and that $v$ is also in $\Gamma_u(S)$. Hence, 
    \begin{align}
        |\Gamma_u(S)| \geq c(1-\varepsilon)|S| - \varepsilon c |S| = (1-2\varepsilon) |S|.
    \end{align}
\end{proof}

We note that if $\varepsilon \to 0$, then similar to lossless vertex expansion, we additionally obtain \textit{lossless unique neighbor expansion}. It is a remarkable fact that such lossless expanders exist and in fact they are quite generic. In the subsequent sections we will show this fact, and show how diffusion codes attain this property on a sub-extensive scale.

\subsection{Linear Classical and Quantum CSS Codes}

\begin{defn} [Classical linear codes]
    A classical linear code $\mathsf{C}$ of $n$ bits and $m$ checks is defined by a parity check matrix $\mathbf{H} \in \mathbb{F}_2^{m\times n}$ with $\text{rank}(\mathbf{H}) \leq n$. The codewords of the code are the elements of the subspace defined by $\ker (\mathbf{H})$. We also define the number of logical bits $k$ as the size of the logical subspace, $k = \dim(\ker(\mathbf{H}))$. The distance of the code $d$ is defined by the minimum weight of a non-zero element of $\ker (\mathbf{H})$. 
    
    Given  a code and its parity check matrix $\mathbf{H}$, we say that it is an $[n,k,d]$ code if it is defined in the space $\mathbb{F}_2^n$, the dimension of the kernel of $\mathbf{H}$ is of dimension $k$ and the code has distance $d$.
\end{defn}

\begin{defn} [Low Density Parity Check Codes] \label{def:ldpc}
    Given a code $\mathsf{C}$ with parity check matrix $\mathbf{H} \in \mathbb{F}_2^{m\times n}$, if $\sum_\ell (\mathbf{H})_{\ell, r} \leq \wbit$ and $\sum_r (\mathbf{H})_{\ell, r} \leq \wcheck$ for constants $\wbit, \wcheck$, then the code defined by $\mathbf{H}$ is a low density parity check (LDPC) code. We say that $\wbit$ and $\wcheck$ are the bit and check degree respectively.
\end{defn}

\begin{remark}
    We shall restrict ourselves to binary codes.
\end{remark}

A classical linear code is defined by a basis of $k$ bit strings of length $n$ (i.e. codewords) which are the unique elements of $\mathbb{F}_2^n$ that satisfy a set of parity checks. These parity checks may be organized into a parity check matrix $\mathbf{H}$ of rank $k$ and column dimension $n$, such that the $k$ codewords form a basis of the kernel of $\mathbf{H}$. Given some arbitrary word $\mathbf{x}$, we may determine its \textit{syndrome}, by calculating $\mathbf{s} = \mathbf{H}\mathbf{x}$. The set of words within the logical subspace are the unique words which satisfy all the parity checks, that is if $\mathbf{z} \in \ker(\mathbf{H})$, then $\mathbf{H}\mathbf{z} = 0$. An arbitrary word may be decomposed into a codeword $\mathbf{z}$ and an error $\mathbf{e} \in \mathbb{F}_2^n$. The task given some word $\mathbf{x}$ is to determine this decomposition and therefore recover the correct codeword. Given the distance of the code $d$, it is in principle possible to do this as long as $|\mathbf{e}| < d / 2$.

Linear codes admit a nice bipartite graph representation. We call the bipartite graph associated to a given code, the \textit{Tanner graph} of the code.

\begin{defn} [Tanner graph] \label{def:tanner graph}
    Given a code defined by a parity check matrix $\mathbf{H} \in \mathbb{F}_2^{m\times n}$, the Tanner graph of the code is the bipartite graph $\mB = (L,R,E)$ which has $\mathbf{H}$ as its biadjacency matrix.
\end{defn}

The Tanner graph $\mB = (L,R,E)$ of a code describes the connectivity between bits and checks. By our convention, the left vertices of the Tanner graph correspond to bits and the right vertices correspond to checks. An edge connects a check to a bit if that check has weight on that bit. For an LDPC code of bit and check degree $\wbit, \wcheck$, the Tanner graph is left / right bounded with left side vertex degree bound $\wbit$ and right side vertex bound $\wcheck$ (\cref{def: left right boundedness}). 

The reverse is also true. Given some bipartite graph $\mB = (L,R,E)$, one may use it to define a code by taking its biadjacency matrix as the parity check matrix of the code. This motivates the study of relating useful code properties to graph theoretic properties. In particular interest is the case where $\mB$ is an expander graph (\cref{def:leftright expansion}).

\begin{defn}[Expander codes ([Definition 11.3.2 in \cite{essentialcodingtheory})] \label{def:expander codes}
If the Tanner graph $\mB$ of a code $\mathsf{C}$ is a left vertex expander according to \cref{def:leftright expansion}, then $\mathsf{C}$ is said to be an \textit{expander code}.
\end{defn}

Expander codes have many useful properties. First, strong enough expansion guarantees a lower bound on the distance.

\begin{lem} [Expansion guarantees lower bound on distance] \label{lemma:expansion distance}

Let $\mathsf{C}$ be an expander code with $n$ bits and $m$ checks and bit degree $\wbit$. If the Tanner graph $\mB =(L,R,E)$ of $\mathsf{C}$ is a lossless $(\delta(n), \wbit(1-\varepsilon))$ expander with $\varepsilon < 1/2$, then $\mathcal{C}$ has distance at least $\delta(n)+ 1$.
    
\end{lem}

\begin{proof}
    We may prove this by contradiction. Let us assume that the distance $d \leq \delta(n)$. Then by the definition of distance, there exists a nonzero codeword $\mathbf{z}$ with weight $|\mathbf{z}| \leq \delta(n)$. Let $S$ be the coordinates of nonzero elements of $\mathbf{z}$. By \cref{lemma:unique neighbor exp}, $|\Gamma_u (S)| \geq \wbit (1-2\varepsilon) |S| > 0$. Thus, $\Gamma_u(S)$ is non-empty and there exists at least one $r\in R$ which connects only once into the set $S$. This $r$ is a parity check and by the definition of $S$, all other neighbors of $r$ in $L$ correspond to bits in $\mathbf{z}$ which are 0. Thus, the parity check corresponding to $r$ must necessarily be unsatisfied, which is a contradiction.
\end{proof}

Second, strong enough expansion also guarantees a property known as \textit{confinement}.

\begin{defn} [Confinement] \label{def:confinement}
    Let $\mathbf{H}$ be the parity check matrix of a code $\mathsf{C}$ with $n$ bits and $m$ checks. The code $\mathsf{C}$ is said to exhibit $(\delta(n), \gamma)$ confinement if $\forall \mathbf{e} \notin \ker (\mathbf{H})\ |\ |\mathbf{e}| \leq \delta(n)$, $|\mathbf{H} \mathbf{e} | \geq \gamma |\mathbf{e}|$.
\end{defn}

\begin{lem} [Expansion guarantees confinement] \label{lemma:expansion confinement}
    Let $\mathsf{C}$ be an expander code with $n$ bits and $m$ checks and bit degree $\wbit$. If the Tanner graph $\mB =(L,R,E)$ of $\mathsf{C}$ is a lossless left $(\delta(n), \wbit(1-\varepsilon))$ vertex expander with $\varepsilon < 1/2$, then $\mathsf{C}$ exhibits $(\delta(n), \wbit(1-2\varepsilon))$ confinement.
\end{lem}

\begin{proof}
    We prove this by contradiction. Let $\mathbf{e}$ be an error with weight $|\mathbf{e}| \leq \delta$ and $|\mathbf{H} \mathbf{e}| \leq \wbit(1-2\varepsilon)|\mathbf{e}|$. Let $S$ be the coordinates of $\mathbf{e}$ which are $1$, which means $|\mathbf{e}| = |S|$. By \cref{lemma:unique neighbor exp}, $|\Gamma_u(S)| \geq \wbit(1-2\varepsilon) |S|$. As we saw in the distance proof above, the elements of $\Gamma_u(S)$ correspond to checks which are necessarily unsatisfied, and so we have a contradiction.
\end{proof}

A quantum code is defined similar to a classical code in that we again have a set of $k$ codewords of dimension $n$ which form a basis of space that is kernel of some parity check matrix. However, the elements of each codeword are now instead drawn from $\mathbb{F}_2 \times \mathbb{F}_2$, owing to the fact that the parity checks of the code now correspond to $n$-qubit Pauli operators. 

A CSS code is a special type of construction of a quantum code in which we have two separate classical codes that obey a nice structure.  

\begin{defn}[Quantum CSS Codes]\label{def:quantum css}
    A quantum CSS code is defined by a pair of classical codes $\mathsf{C}_X$ and $\mathsf{C}_Z$ which obey the relation $\mathsf{C}_Z^\perp \subseteq \mathsf{C}_X$, where $\mathsf
{C}_Z^\perp$ corresponds to the words which are orthogonal to the codewords of $\mathsf{C}_Z$. This condition manifests as a constraint on the parity check matrices of $\mathsf{C}_X$ and $\mathsf{C}_Z$, that is $\mathbf{H}_X \cdot \mathbf{H}_Z^T = 0$. We say that $\mathsf{C}_{X(Z)}$ is the code which corresponds to the $X(Z)$ sector of the quantum code.
\end{defn}

Constructing CSS codes with favorable properties is a nontrivial task. However, quantum expander codes with favorable properties such as $k,d$ growing with $n$ and a finite bit, check weight may be constructed using products of classical expander codes. One construction is given by the \textit{hypergraph product}.

\begin{defn} \label{def:hypergraph product}
    (Hypergraph product CSS code). Given a bipartite graph $\mathcal{B} = (L,R,E)$, one may construct a quantum CSS code via a hypergraph product:
    \begin{enumerate}
        \item The bipartite graph $\mathcal{B}_X$ corresponds to the Tanner graph of the $X$ sector of the code. The left set of vertices of $\mathcal{B}_X$ is $L^2 \cup R^2$ and its right set of vertices is $L \times R$. That is, $(i,j) \in L^2 \cup R^2$ if $i,j \in L$ or $i,j \in R$, and $(i,j) \in L \times R$ if $i \in L$ and $j \in R$. If a left vertex $\alpha a \in A^2$, then
        \begin{align}
            \Gamma(\alpha a) = \{ \alpha \beta \in A \times B, (a, \beta) \in E \}
        \end{align}
        and if a left vertex $b\beta \in B^2$, then
        \begin{align}
            \Gamma(b \beta) = \{ \alpha \beta \in A \times B, (\alpha, b) \in E \}.
        \end{align}
        \item The bipartite graph $\mathcal{B}_Z$ corresponds to the Tanner graph of the $Z$ sector of the code. The left set of vertices of $\mathcal{B}_Z$ is $L^2 \cup R^2$ and its right set of vertices is $R \times L$. That is, $(i,j) \in L^2 \cup R^2$ if $i,j \in L$ or $i,j \in R$, and $(i,j) \in R \times L$ if $i \in R$ and $j \in L$. If a left vertex $\alpha a \in A^2$, then
        \begin{align}
            \Gamma(\alpha a) = \{ ba \in B \times A, (\alpha,b) \in E \}
        \end{align}
        and if a left vertex $b\beta \in B^2$, then
        \begin{align}
            \Gamma(b \beta) = \{ ba \in B \times A, (a,\beta) \in E \}.
        \end{align}
    \end{enumerate}

\end{defn}

\begin{lem} [Properties of hypergraph product code \cite{Leverrier_2015}]
    The quantum code defined from the hypergraph product of a $[n_A, k_A, d_A]$ classical code $\mathsf{C}_A$ and another $[n_B, k_B, d_B]$ classical code $\mathsf{C}_B$ has a number of qubits $n = n_A^2 + n_B^2$, a number of logical qubits $k\geq (n_A - n_B)^2$ and distance $d = \min (d_A,d_B)$.
\end{lem}

\begin{proof}
    We refer the reader to \cite{Leverrier_2015} for proof of this lemma.
\end{proof}

\subsection{Existence and Construction of Lossless Expanders}

\begin{theo}[Existence of lossless expanders \cite{essentialcodingtheory}]\label{theorem:losslessexpander} 
    Consider $\mathcal{B} = (L,R,E)$ with $|L| = n$, $|R| = m$, with $m = \mathcal{O}(n)$, and the left vertex degree be $c$. $\mB$ is constructed such that $\forall \ell \in L$, each of the $c$ edges emanating from it connects uniformly at random to a neighbor $r \in R$. Then with finite probability, $\mathcal{B}$ is a left $(\delta(n), \gamma)$ vertex expander with $\delta(n) = \delta n,\ \gamma > c(1-\varepsilon (c))$, where $\varepsilon(c) \to 0$ as $c \to \infty$. 
\end{theo}

\begin{proof}
Let $S \subseteq L$ be a subset such that $1 \leq |S| \leq \delta n$. For every $\ell \in S$, $\Gamma(\ell)$ consists of $c$ elements of $R$ chosen uniformly at random. Let $r_1, \dots, r_{c|S|}$ be the neighbors of $S$ in some arbitrary order. A choice $r_j$ is a "repeat" if $\exists i \leq j$ such that $r_i = r_j$. If the total number of repeats is less than $\varepsilon c |S|$, then $|\Gamma(S)| \geq c(1-\varepsilon )|S|$.  

If we reveal the labels $r_i$ sequentially, then the probability that $r_i$ is a repeat is at most $\frac{i-1}{m} \leq \frac{c|S|}{m}$. This follows from the fact that there are $m$ choices for $r_i$ and the repeat probability is maximized if all of the previous choices are distinct. The second bound is from the fact that $i < c|S|$.

From this, we obtain the probability of having more than $\varepsilon c |S|$ repeats in the list ($r_1, \dots, r_{c|S|})$:
\begin{align} \label{eq:repeatsinS_expander}
    \mathbb{P} [ > \varepsilon c |S|\ \text{repeats}] \leq \begin{pmatrix}
        c|S| \\ \varepsilon c |S|
    \end{pmatrix} \left( \frac{|S|c}{m} \right)^{\varepsilon c |S|} \leq \left( \frac{e|S|c}{\varepsilon m} \right)^{\varepsilon c |S|}.
\end{align}

The first bound follows from having $\begin{pmatrix}
        c|S| \\ \varepsilon c |S|
    \end{pmatrix}$ to arrange $\varepsilon c |S|$ repeats while integrating over the possibilities for the remaining nodes. The second bound follows from the bound on binomial coefficients. The probability that \textit{any} arbitrary subset $S$ with $|S| \leq \delta n$ fails to expand is exponentially suppressed in the size of the subset. 

By a union bound, we can upper bound the probability that a single set $S \subset L$ with $1 \leq |S| \leq \delta n$ fails to expand (has $|\Gamma (S)| < c(1-\varepsilon) |S|$).
\begin{align}
    \mathbb{P}(\text{fail}) &\leq \sum_{|S| = 1}^{\delta n} \begin{pmatrix}
        n \\ |S|
    \end{pmatrix} \left( \frac{e|S|c}{\varepsilon m} \right)^{\varepsilon c |S|} \\
    &\leq \sum_{|S| = 1}^{\delta n} \left(\frac{e n}{|S|} \right)^{|S|}  \left( \frac{e|S|c}{\varepsilon m} \right)^{\varepsilon c |S|} \\
    &= \sum_{|S| = 1}^{\delta n} \left(\frac{e n}{|S|}  \left( \frac{e|S|c}{\varepsilon m} \right)^{\varepsilon c }\right)^{|S|}.
\end{align}

The terms in the summation represent the number of ways to create a subset of size $|S|$. In the third line, we have written the summation as a geometric series in $|S|$. If the argument of geometric series is less than 1/2, the entire summation becomes bounded away from 1. Futhermore, we note that the entire argument is either a strictly increasing or strictly decreasing function of $|S|$, depending on the sign of $\varepsilon c - 1$. Therefore, we may simply bound the first and last terms of the sum. 

We then seek to attain a condition on $c$ such that the following inequalities are true:
\begin{align}
    en (\delta n)^{\varepsilon c - 1} \left( \frac{e c}{\varepsilon m}\right)^{\varepsilon c} & \leq \frac{1}{2} \\
    en \left( \frac{e c}{\varepsilon m}\right)^{\varepsilon c} & \leq \frac{1}{2}.
\end{align}

Algebraic manipulations of these inequalities yield:
\begin{align}
    c \geq \frac{1}{\varepsilon} \frac{1 + \log_2 \left( \frac{e n }{\delta n} \right)}{\log_2 \left(\frac{\varepsilon m}{\delta n e c} \right)} \\
    c \geq \frac{1}{\varepsilon} \frac{1 + \log_2 \left( en \right)}{\log_2 \left(\frac{\varepsilon m}{ e c} \right)}
\end{align}

Since we have demanded that $m = \mathcal{O}(n)$, it is possible to simultaneously satisfy both of these inequalities for any value of $\varepsilon$ between 0 and 1, with a $c$ which does not grow with $n$. If the choice of $c$ satisfies these inequalities, then the graph as constructed will, with nonzero probability in the large $n,m$ limit, be a lossless expander. 

\end{proof}

In the construction above, we showed that when each of the $c$ neighbors of every left vertex are chosen uniformly at random from all of the right vertices, depending on $c$ and $\varepsilon$, we may obtain a lossless expander. The construction above however places no guarantee on right vertex degree. This may be rectified by adding vertices on the right side \cite{essentialcodingtheory}. 

\begin{lem}[Right vertex correction \cite{essentialcodingtheory}] \label{lemma:vertexcorrection}
    Let $\mathcal{B} = (L,R,E)$ be a lossless $(\delta, \gamma)$ expander constructed according to theorem \ref{theorem:losslessexpander}. Let $|L| = n$, $|R| = m$ and the left vertex degree be $c$. Then, there exists another bipartite graph $\mathcal{B}' = (L,R',E')$, with $|R'| = m$ that is also a $(\delta, \gamma)$ bipartite expander such that 
    \begin{itemize}
        \item $m \leq m' \leq 2m$
        \item Every right vertex in $\mathcal{B}'$ has degree at most $\lceil \frac{nc}{m} \rceil$.
    \end{itemize}
\end{lem}

\begin{proof}
    Let us define
    \begin{align}
        d = \lceil \frac{nc}{m} \rceil.
    \end{align}

    For every vertex $r \in R$, let $d_r$ be its degree. For each $r \in R$, we add $\lceil \frac{d_r}{d} \rceil$ vertices to $R'$. That way, every vertex in $R$ corresponds to $\geq 1$ vertex in $R'$. The edges incident on $r \in R$ are divided evenly among the corresponding vertices $r' \in R'$. This causes all of the vertices in $R'$ which correspond to $r \in R$ to have degree $d$ except for at most 1, which has degree $\leq d$. Furthermore, we may show that $m \leq m' \leq 2m$:
    \begin{align}
        m' - m = \sum_{r \in R} \left( \lceil \frac{d_r}{d} \rceil - 1\right) \leq \sum_{r\in R}\frac{d_r}{d} = \frac{nc}{d} \leq m.
    \end{align}
\end{proof}

The above method to attain a lossless expander requires the addition of new vertices which may or may not be desirable. A more elegant solution is given by Gallager codes which are the paradigmatic LDPC codes. These codes are probabilistically constructed and with high probability, their Tanner graphs are of fixed left and right degrees. \cite{gallagerLDPC}.

We reiterate the construction of Gallager codes here as follows: Let $\wbit$ denote the desired bit degree of the code and let $\wcheck$ be the desired check degree. We create $n\wbit$ nodes and number them from 1 to $n\wbit$. For each bit, there will be $\wbit$ corresponding nodes within this arrangement and so given a node $i$, $\lceil \frac{i}{\wbit} \rceil$ yields the bit to which it belongs. Let us call the group of nodes corresponding to a bit, the "socket" of the bit. We then take adjacent groups of $\wcheck$ of these nodes. There will be $m$ of these groups and we say that the nodes contained within these groups are in the "socket" of the corresponding check. Integer division of the node labels within socket $j$ yields the bits which participate in check $j$. Now, in order to generate the Tanner graph of a code within the Gallager LDPC ensemble, we perform a complete random permutation of all of the $n \wbit$ nodes. This will completely randomize which nodes are within which socket. Then, following integer division by $\wbit$ of the node labels, we will obtain all of the desired connections between each check and bit for this randomized code. With high probability the Tanner graph of this code is expanding (defined below) and the code rate $k/n = 1 - \wbit/\wcheck$ \cite{Richardson_Urbanke_2008, essentialcodingtheory,gallagerLDPC}.

\begin{lem} 
    [Gallager codes are lossless expanders \cite{Richardson_Urbanke_2008, essentialcodingtheory}]\label{lemma:gallager} Consider the graph constructed in the following manner: Given a random matching $\mG$ between $n\wbit = m\wcheck$ nodes, the graph that results from collapsing groups of $\wbit$ adjacent nodes on the left side and $\wcheck$ adjacent nodes on the right side is a lossless expander. 
\end{lem}

\begin{proof}
    Let $\mB = (L,R,E)$ be the Tanner graph of the code, that is, the graph obtained after collapsing nodes in the sockets in $\mathcal{G}$. Let $S \subseteq L$ be a subset such that $1 \leq |S| \leq \delta n$. Let $r_1, \dots, r_{\wbit|S|}$ be the neighbors of $S$. We again define $r_j$ to be a repeat if $\exists i < j$ such that $r_i = r_j$. We again would like to consider the process of revealing the labels $r_i$ sequentially. However, each of the $r_i$ can be identified with an edge in $\mathcal{G}$ and the identity of $r_i$ corresponds to the socket to which that edge connects. 

    We may ask, what is the probability that $r_i$ is \textit{not} a repeat. This is the probability that $\nexists j < i| r_j = r_i$. Identically, this is the probability that as we reveal the identities of the labels, none of them connect to the same socket as $r_i$. We may write this in terms of conditional probabilities:
    \begin{align}
        \mathbb{P} (r_i\ \text{not a repeat}) = \mathbb{P}(r_1 \neq r_i) \mathbb{P}(r_2 \neq r_i | r_1 \neq r_i) \dots
    \end{align}

    When we reveal the first label $r_1$, there are $md$ choices of nodes for $r_1$ to connect into in $\mathcal{G}$. $\wcheck$ of those choices will be in the same socket as $r_i$. Therefore, $\mathbb{P}(r_1\neq r_i) = \frac{m\wcheck - \wcheck}{m\wcheck}$. When we reveal $r_2$, $r_2$ now has $m\wcheck-1$ choices. Again, $\wcheck$ of them are contained within the socket of $r_i$, and from the conditional, they are unoccupied. Therefore, $\mathbb{P}(r_2 \neq r_i | r_1 \neq r_i) = \frac{m\wcheck-1-\wcheck}{m\wcheck-1}$. The product of the conditionals above amounts to
    \begin{align}
        \mathbb{P} (r_i\ \text{not a repeat}) &= \prod_{j=0}^{i-1} \frac{m\wcheck-\wcheck-j}{m\wcheck-j} \\
        &= \prod_{j=0}^{i-1} \left(1 - \frac{\wcheck}{m\wcheck-j} \right).
    \end{align}

    With this, we obtain the probability that $r_i$ \textit{is} a repeat:
    \begin{align}
        \mathbb{P} (r_i\ \text{is a repeat}) &= 1 - \prod_{j=0}^{i-1} \left(1 - \frac{\wcheck}{m\wcheck-j} \right) \\
        &\leq \sum_{j=0}^{i-1} \frac{\wcheck}{m\wcheck-j} \\
        &\leq \sum_{j=0}^{i-1} \frac{\wcheck}{m\wcheck} \\
        &= \frac{i-1}{m} \\
        &\leq \frac{\wbit|S|}{m}.
    \end{align}
    The first bound above follows algebraically (may also be interpreted as a union bound) and the second bound follows from $j$ merely decreasing the denominator of each term in the summation. The last bound follows from $i-1 \leq \wbit|S|$. We obtain an upper bound on the repeat probability which is the same as that in the proof of theorem \ref{theorem:losslessexpander}. Subsequently, the remainder of the proof is identical.
    
\end{proof}

\subsection{Markov chains and Mixing Times}

In proving one of the main results of this work, we will make use of the mixing time of Markov chains. By showing that a Markov chain has evolved for a time greater than its mixing time, we will make use of the stationary distribution of the Markov chain in order to bound quantities of interest. Here we provide some definitions used in later parts of this work. For more background on Markov chains and mixing times, we refer the reader to \cite{markovchains}.

\begin{defn} [Finite Markov Chains] \label{def:markov chain}
    A finite Markov chain defines an evolution on a space of configurations $\mathcal{X}$. If $x_t$ is the configuration of the Markov chain at time $t$, the next configuration $x_{t+1}$ is chosen probabilistically based on a transition matrix $P_t$ which depends only on the current state $x_t$.
\end{defn}

\begin{defn}  [Total Variation Distance] \label{def:tvd}
    Let $\mu, \nu$ be two probability distributions defined on a space of configurations $\mathcal{X}$. The total variation distance is defined by
    \begin{align}
    ||\mu - \nu ||_{\rm TVD} = \max_{A \subseteq \mathcal{X}} |\mu(A) - \nu(A)|.
    \end{align}
\end{defn}

\begin{defn} [Mixing Time] \label{def:mixing time}
    For any Markov chain, the $\epsilon$-mixing time $t_{\rm mix}(\epsilon)$ is given by the first time $t$ at which the total variation distance (def. \ref{def:tvd}) between the current probability distribution of the Markov chain and the stationary distribution becomes less than or equal to $\epsilon$.
\end{defn}

\section{Proof of Main Results}

In the following, we present the proof of our main results summarized in \cref{sec:summary_results}. We begin by showing the existence of bipartite smaller set lossless expanders via a simple random geometric construction in \cref{sec:proof_existence}. 
We then proceed to show that diffusion codes on the cycle graph are smaller set lossless expanders with high probability.
We then discuss hypergraph products of diffusion codes and their smaller set (co-)boundary expansion. 
Finally, we discuss the the implications of the smaller set expansion, which include the existence of a linear-time local decoder, single-shot error correction, and passive memory.

Though in the definitions (\cref{def:vertexexpander}) we defined expander graphs and codes with $\delta(n)$ left ambiguous, in the above and following, we reiterate that we use the term \textit{smaller set expansion} to refer to expander graphs where $\delta(n) = o(n)$.

\subsection{Existence of Smaller Set Lossless Expanders\label{sec:proof_existence}}


Our proof for the existence of bipartite smaller set lossless expanders is directly inspired from the proof of existence of `regular' bipartite lossless expanders as presented in \cref{theorem:losslessexpander}. In this case, each left vertex is connected to a random subset of $c$ right vertices. In the following, we consider a simple generalization of this idea, where we define for each left vertex $\ell\in L$ a candidate neighbor set, $\Upsilon_\ell\subset R$. Then, as long as each candidate vertex set is sufficiently large, we can guarantee smaller set lossless expansion from left to right vertices.

Importantly, the below theorem does not place any restriction onto the candidate sets (they could, in fact, all overlap in the worst case). A natural case would be, for example, to place left and right vertices onto some manifold (e.g. $\mathbb{R}^d$), and choosing the candidate set for a given left vertex as all right vertices within a certain distance.

\begin{theo} \label{theorem:sublinearlossless}
    (Existence of smaller set expanders). Consider a bipartite graph $\mathcal{B} = (L,R,E)$ with $|L| = n, |R| = m = \mathcal{O}(n),$ and with left-degree $c$ constructed as follows. Define for each $\ell\in L$ a candidate neighbor set $\Upsilon_\ell \subset R$, and connect $\ell$ to $c$ vertices from $\Upsilon_\ell$ chosen uniformly at random.

    Now, $\forall \varepsilon > 0$, if $\forall \ell, |\Upsilon_\ell| > \delta$ with  $\delta = \Omega(n^\beta)$ and $\beta > (\varepsilon c)^{-1}$, then with probability $q(n) > 0$, $\mathcal{B}$ is a left $(\delta, \gamma)$ smaller set vertex expander (\cref{def:vertexexpander}) with $\gamma > c(1-\varepsilon)$, where $q(n) \to 1$ as $\abs{L}\to 0$.     
\end{theo}


\begin{proof}
    Let $S \subset L$ be a subset such that $1 \leq |S| \leq \delta(n)$. We again let $r_1, \dots, r_{c|S|}$ be the neighbors of $S$, with the definition of $r_j$ being a repeat being the same as before. In the worst case scenario, all of the candidate neighbor sets for each $\ell \in S$ exactly overlap, as this maximizes the repeat probability. We additionally restrict the candidate neighbor sets to be of size $\delta(n)$. Then, the probability that $r_i$ is a repeat is at most $\frac{i-1}{\delta(n)} \leq \frac{c|S|}{\delta(n)}$, where the bounds come from the scenario where each of the $i-1$ choices before $i$ were distinct, and $i-1 < c|S|$. We then obtain the probability of having more than $\varepsilon c |S|$ repeats in the list $(r_1,\dots,r_{c|S|})$:
    \begin{align} \label{eq:sublinear bound}
        \mathbb{P}[> \varepsilon c |S| \text{ repeats}] \leq \left( \frac{e|S|c}{\varepsilon \delta(n)} \right)^{\varepsilon c |S|}.
    \end{align}

    Here, the bounds come from the same arguments as in \cref{eq:repeatsinS_expander}. We now seek to bound the probability that in $\mathcal{B}$, there exists a single set $S \subset L$ with $1 \leq |S| \leq \delta(n)$ fails to expand. However, naively performing a union bound over all subsets will not yield the desired result. The reason is that in the earlier proof, the probability of a subset failing to expand was exponentially suppressed in $n$, which competed with the exponential number of subsets of size $|S|$. However, if $\delta(n)$ is a sublinear function of $n$, then the overall number of subsets will overwhelm the probability of the subset to fail to expand. 

    Instead, we will use the fact that any subset $S \subset L$ may be decomposed into a union of \textit{neighbor-connected subsets}. By neighbor connected, we mean that left vertices which are connected through a neighbor in $R$. That is, $v,v' \in L$ are neighbor connected iff $\Gamma(v) \cap \Gamma(v') \neq \empty$. We may write the decomposition of $S$ as
    \begin{align}
        S = \cup_{i=1}^k \mathcal{S}_k,
    \end{align}
    where $\mathcal{S}_k \subset L$ is neighbor connected and each of the $\mathcal{S}_k$ are disjoint. By definition, since the $\mathcal{S}_k$ are disjoint, their neighbors do not overlap. Subsequently,
    \begin{align}
        |\Gamma(S)| = \sum_k |\Gamma (\mathcal{S}_k)|.
    \end{align}

    It then suffices to bound the probability that in $\mathcal{B}$, there exists a single neighbor-connected set $\mathcal{S} \subset L$ with $1 \leq |{\mathcal{S}}| \leq \delta(n)$ fails to expand. 
    
    The number of such sets is \cite{Fawzi_2018}
    \begin{align}
        |C_{\mathcal{S}}(\mathcal{B})| \leq n\Phi^{|{\mathcal{S}}|},
    \end{align}
    where 
    \begin{align}
        \Phi \leq (c^2-1) \left( 1 + \frac{1}{c^2- 2} \right)^{c^2- 2}.
    \end{align}
    The benefit of this strategy is clear. The number of neighbor connected sets of size $|\mathcal{S}|$ is only \textit{polynomial} in $n$ rather than exponential. The union bound may now be written as
    \begin{align} \label{eq:sum bound}
        \mathbb{P}(\text{fail}) &\leq \sum_{|\mathcal{S}| = 1}^{\delta(n)} n \Phi^{|\mathcal{S}|} \left( \frac{e|\mathcal{S}|c}{\varepsilon \delta(n)} \right)^{\varepsilon c |\mathcal{S}|} \\
        &= \sum_{|\mathcal{S}| = 1}^{\delta(n)} \left( n^{1/|\mathcal{S}|} \Phi \left( \frac{e|\mathcal{S}|c}{\varepsilon \delta(n)} \right)^{\varepsilon c} \right)^{|\mathcal{S}|}
    \end{align}

    Defining the argument in the parenthesis to be $q$, the above expression is upper bounded by
    \begin{align}
        &= \sum_{|\mathcal{S}| = 1}^{\delta(n)} q^{|\mathcal{S}|} \leq \frac{q}{1-q},
    \end{align}
    as long as $q < 1$ for all $|\mathcal{S}|$. If we allow $\delta(n) \sim n^\beta$, where $0 < \beta < 1$, then $q \sim n^{\frac{1}{|\mathcal{S}|\varepsilon c} - \beta}$, which for large $n$, vanishes if $\beta > \frac{1}{|\mathcal{S}|\varepsilon c}$. This inequality is most difficult to satisfy when $|\mathcal{S}| = 1$, and so we obtain the condition:
    \begin{align}
        \beta > \frac{1}{\varepsilon c}.
    \end{align}
    When this condition is satisfied, then the above bound on the probability for $\mathcal{B}$ to fail to expand vanishes in the large $n$ limit. 
    
\end{proof}

\begin{remark}
    If one proceeds as we did in proving theorem \cref{theorem:losslessexpander} and bound each term in the summation in \cref{eq:sum bound}, then following some algebraic manipulations and in the large $n$ limit, one will obtain the inequality:
\begin{align}
    \varepsilon c \geq \frac{1}{|\mathcal{S}|} \frac{\log_2 m}{\log_2 \delta(n)}.
\end{align}

Since $m = \mathcal{O}(n)$, if we would choose $\delta (n) \sim \text{polylog}(n)$, then we see that to guarantee smaller set expansion we need to chose the left degree as $c \sim \frac{\log n}{\log \text{polylog} n}$. This is extremely slow growing but nevertheless diverges at large $n$.
\end{remark}

As mentioned earlier, a natural case to consider is that we take both left and right vertices to be placed on some manifold, with Euclidean space $\mathbb{R}^d$ being of particular interest, provided, every left vertex has a candidate neighbor set of sufficient size. For many regular placings of left and right neighbors, it is easy to estimate the size of the candidate the neighbor candidate set if one considers simply all right vertices within a given distance.
Even for both types of vertices randomly placed in $\mathbb{R}^d$, the resulting random geometric graph is a smaller set lossless expander.

\begin{lem} [Smaller set expanding random geometric graphs] \label{lemma:geometricsublinear}

    Given a unit volume in the space $\mathbb{R}^d$, consider a bipartite graph $\mB = (L,R,E)$ with $|L| = n, |R| = m = \mathcal{O}(n),$ and with left degree $c$ constructed as follows. Place each of the left and right vertices uniformly at random on the manifold. Define for each $\ell \in L$ a candidate neighbor set $\Upsilon_\ell \subset R$ as the subset set of right vertices contained within a volume $V_\ell$ surrounding $\ell$. Then, connect $\ell$ to $c$ vertices from $\Upsilon_\ell$ chosen uniformly at random.

    Now, $\forall \varepsilon > 0$, if $\forall \ell$, $V_\ell = \mO (m^{\beta - 1})$ with $\beta > (\varepsilon  c)^{-1}$, then almost surely $\forall \ell$ $|\Upsilon_\ell| > \delta$ with $\delta = \Omega(n^\beta)$  and $\mB$ is a left $(\delta(n), \gamma)$ smaller set vertex expander (\cref{def:vertexexpander}) with $\gamma > c(1-\varepsilon)$.
     
\end{lem}

\begin{proof}
    Consider $\ell \in L$ and the volume $V_\ell = \alpha m^{\beta - 1}$ surrounding it, where $\alpha$ is some positive constant. We would like to know whether $|\Upsilon_\ell| \geq m^\beta\ \forall \ell \in L$. This requires that $\forall \ell$, $V_\ell$ must contain $\geq m^\beta$ right vertices. If this condition is met, then the proof of the construction in the lemma being a sublinear lossless expander follows identically to that of  \cref{theorem:sublinearlossless}.
    
Since the right vertices are placed iid uniformly at random in the unit volume, the probability that any given right vertex is in $V_\ell$ is just $V_\ell$.  Let $X = \sum_{r = 1}^{m}\mathbb{1}_{r \in V_{\ell}}$. Then for every $r$, $\mathbb{P}[\mathbb{1}_{r \in V_\ell} = 1] = V_\ell$. If fewer than $m^\beta$ right vertices are in $V_\ell$, then $X < m^\beta$. Note that $\mathbb{E}[X] = mV_\ell = \alpha m^\beta$. Then by the Chernoff bound on binomial distributions,
    \begin{align}
        \mathbb{P}\left(X \leq (1-\epsilon) \alpha m^\beta  \right) \leq \exp \left( - \frac{\epsilon^2 \alpha m^\beta }{2} \right).
    \end{align}

    In particular, if we now set $\alpha \geq 1$ , the probability for there to be arbitrarily fewer than $m^\beta$ right vertices within any $V_\ell$ vanishes exponentially in $m^\beta$. Thus, the candidate neighbor set for every $\ell$ is almost surely large enough. 
 \end{proof}

\begin{remark}
    In the above proof, we showed that the probability for there to be fewer than $m^\beta$ right vertices in a ball around $\ell$ of size $V_\ell$ vanishes. In fact a more general statement is true. Any arbitrary volume of size $\alpha  m^\beta$ in the region will almost surely contain $> m^\beta$ right vertices. This implies that odd scenarios such as clustering of most of the right vertices in a small region of the space is exponentially unlikely. 
\end{remark}

From a physical point of view, the construction of \cref{lemma:geometricsublinear} is of greater interest. Upon constructing the graph, then when viewed as a Tanner graph, the right vertices define checks that involve bits that are located within a connected volume of size $m^{\beta - 1}$. Identically, this means that in $\mathbb{R}^d$, each check only involves bits that are at most a distance $\mathcal{O}\left( m^{\frac{\beta-1}{d}}\right)$ apart. This quantity of course diverges but much more slowly than the linear divergence that would emerge from attempting to embed a true expander in $\mathbb{R}^d$.

The construction presented above is very general, and using it, it is easy to construct lossless smaller set expanders where the edge length of the graph is bounded relative to some metric. However, it is not symmetric between left and right vertices: for example while the left-degree of the graph is bounded by the constant $c$, the right degree might not be bounded in general. Furthermore, it is not obvious in which cases the graph will both be left and right expanding.
Hence, when taking this construction to define a linear code, we do in general not know whether this code will be LDPC. Furthermore, if one would try to take the hypergraph product of such codes to define a quantum LDPC codes, the right expansion properties become important even to lower bound the distance and both the bit and check degree. In theory, one could use \cref{lemma:vertexcorrection} to alleviate this, however this will yield an asymmetric construction. Because of this, it is clearly desirable to obtain a construction more symmetric between the left and right vertices. We provide such a construction in the next section, which can be viewed as a generalization of the configuration model.

\subsection{Diffusion Codes}\label{sec:construction}

In this section, we define in full detail the diffusion codes which we introduced informally in \cref{sec:summary_diffusion_codes}, and derive that if placed on the cycle graph, they lead a family of bipartite smaller set lossless expanders.

\begin{figure}
    \centering
    \includegraphics[width=\linewidth]{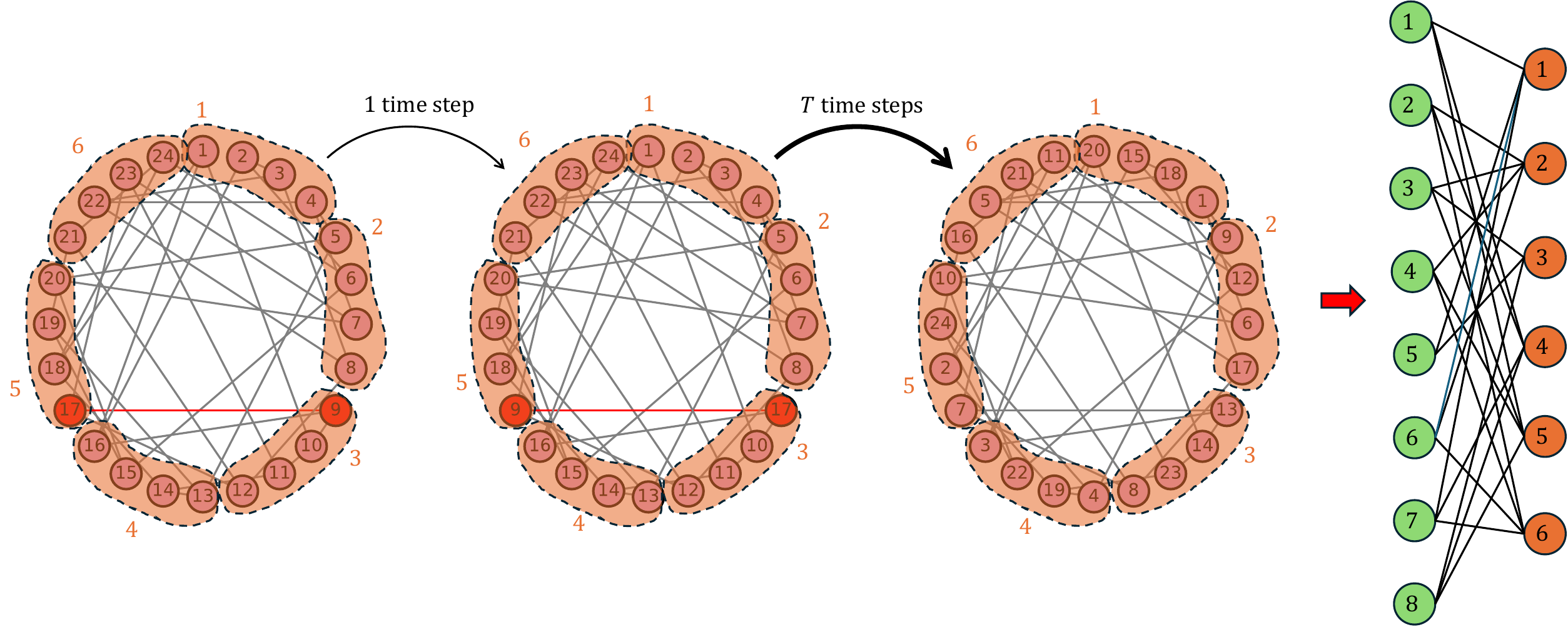}
    \caption{On the left we have a graph with a numbering of the vertices. The interchange process swaps pairs of vertex labels at each time step, the first of which is shown in the middle graph. Here the swap changed the group to which $9$ and $17$ belong. After many such swaps, the final position of the vertex labels, shown in the third graph, enables the creation of the bipartite Tanner graph on the right. The bit identities can be determined by the vertex labels. That is, node $i$ collapses into bit $\lceil i / \wbit \rceil$.
   }
    \label{fig:interchangeprocess}
\end{figure}

We begin with some definitions

\begin{defn}[Interchange Process] \label{def:interchange}
Let $\mG = (V,E)$  be a graph where every vertex is occupied by a distinguishable particle. On each edge, there is a Poisson clock of rate $1$. Each time a clock on an edge rings, the particles on the vertices belonging to that edge are swapped. 
\end{defn}

The interchange process, as defined here, is exactly that of the SWAP network described informally in earlier sections. 

\begin{defn} [Simple Exclusion Process] \label{def:sep}
    By declaring only a subset of particles as visible, and regarding the visible particles as indistinguishable, we obtain the induced process on $\mG$ known as the simple exclusion process (SEP) \cite{markovchains}.
\end{defn}

Running the interchange process for some time $T$ hence generates a permutation of the vertices of $\mG$. We can use this to define a random family of codes as below. Importantly, both the interchange process and the SEP are well studied Markov chains with known bounds on their mixing time \cite{markovchains}, and this can be used to derive rigorous results about the corresponding code family.

\begin{defn}[Diffusion Codes] \label{def:diffusioncodes}
Fix a number of bits $n$, checks $m$, desired bit degree $\wbit$ and desired check degree $\wcheck$. Specify a connected graph $\mG = (V,E_\mG)$ with $|V| = n\wbit = m\wcheck = N$. Next, partition the vertices $V$ into connected groups of $\wcheck$ vertices, with each group numbered from $1$ to $m$, and a map $\phi : V \to [m]$ which yields the group of a vertex. The groups will correspond to the sockets of a check. 
Finally, assign an initial numbering of vertices $\psi_0 : V \to [N]$ such that $\forall v \in V$, if $v$ is in group $j$, that is $\lceil \frac{\psi_0(v)}{\wcheck} \rceil = j$.

Given a time $T \in \mathbb{R}_+$, run the interchange process on $\mG$ with the vertex labels as the particles in order to generate a permutation of the vertex labels $\psi_T$. 
We then define the Tanner graph $\mB$ of the $(n,m,\wbit,\wcheck,T,\mG)$ diffusion code is defined as follows: $\mB = (L,R, E)$ has $|L| = n, |R| = m$. The elements of $E$ are defined according to the bijection $\chi : [N] \to E$ where $\chi(i) = \left( \lceil \frac{i}{\wbit} \rceil, \ \phi ( \psi_T^{-1}(i))\right)$. That is, given a vertex label $i$, we connect the bit $\lceil \frac{i}{\wbit} \rceil$ to the check corresponding to the group to which the vertex with label $i$ belongs.
\end{defn}

\begin{remark}
     There are two relevant graphs here. The first is the graph $\mathcal{G} = (V,E_\mG)$ on which the interchange process occurs. The second is the Tanner graph of the diffusion code $\mathcal{B} = (L,R,E)$. For the sake of clarity, in this section, when we say vertex, we are referring to an element of $V$ on the graph $\mathcal{G}$. We shall refer to the elements of $L$, the left vertices of $\mathcal{B}$ as "bits" and we shall refer to the elements of $R$, the right vertices of $\mathcal{B}$ as "checks." In context, we will make clear when we refer to the edges of $\mathcal{G}$ or the edges of $\mathcal{B}$.
\end{remark}

The diffusion code creation process is illustrated in \cref{fig:interchangeprocess}. We have a graph $\mG$ which we partition, and we label all the vertices. Then we sequentially begin swapping labels between connected vertices, resulting in a permutation of the vertex labels which was a product of strictly local permutations. Then, by looking at each vertex, its label and to which partition it belongs, we can obtain all the edges of $\mB$.

\begin{example} [$T= 0 $ diffusion code]
    Consider a $(n,n,\wbit,\wbit,\mG)$ diffusion code. That is, the bit degree and the check degree are identical and we have the same number of bits and checks. If $T =  0$, then $\psi_T = \psi_0$. However, because $\psi_0$ was defined such that every vertex had a label $i$ which, when $\lceil \frac{i}{\wbit} \rceil$ is calculated, always equals the identity of the partition of that vertex, the resulting Tanner graph yields a one to one matching. From this, we see that $T$ sets a degree of nonlocality to the code. The longer we run the interchange process, the greater the spread in checks to which any particular bit are connected. 
\end{example}

\begin{example}[$T \to \infty$ diffusion code]
If $T\to \infty$, then the permutation generated on $\psi_0$ will approach that of a uniform permutation. In this case, because marginally every vertex label has a uniform probability to be in any partition, this will result in a Gallager code. 
    
\end{example}

Can we obtain any general guarantees on the code properties, such as code expansion for this construction? From the examples above, we see that $T = 0$ is clearly not an expander while $T = \infty$ yields a Gallager code which is a true expander (\cref{lemma:gallager}). 
Intuitively, we might expect see that as the interchange process unfolds, the set of possible neighbors for each bit grows diffusively as $\sqrt{T}$. 
In this case, we expect that after a time $T \sim n^{2\beta}$ for some $\beta > 0$, the average candidate neighbor set for each bit is of size $\mO(n^\beta)$, and from  \cref{theorem:sublinearlossless} we might then expect smaller set expansion on that scale. Importantly however, in \cref{theorem:sublinearlossless}, we assume that the neighboring checks are chosen uniformly at random from the candidate check. We hence have to show that the permutation generated by the interchange process mixes on the relevant scale. This is a challenging problem in general, but below we provide a proof for the diffusion process on the cycle graph.

\subsubsection{Cycle Graph} \label{sec:cyclegraph}

In the following, we consider diffusion codes constructed from the interchange process on the cycle graph $\mathcal{C}_N$. We will show that for appropriate $T$, we obtain bounded degree smaller set lossless expanders with a natural embedding into a 1D lattice, and a bound on the typical size of checks. 


\begin{theo} [Lossless smaller set expansion in diffusion codes on cycle graphs] \label{theorem:diffusioncodecycle}
        Take $\mathcal{C}_N$ the cycle graph of size $N = n \wbit = m \wcheck$, and its obvious partition into contiguous groups of $\wcheck$.  
        For all $\varepsilon, \alpha, \beta > 0$, if $T \propto n^\alpha$ and $(\varepsilon \wbit)^{-1} < \beta < \alpha/2$, then the Tanner graph of the $(n, m, \wbit,\wcheck, \mathcal{C}_N, T)$ diffusion code is, with probability $q(N) > 0$
        , a $(\delta, \gamma)$ lossless smaller set left vertex expander with $\delta \propto n^{\beta}$ and $\gamma > \wbit(1-\varepsilon)$. Further, $q(N)\to 1$ as $N\to\infty$.
\end{theo}

\begin{remark}
    The lossless smaller set expansion of the Tanner graph of diffusion codes implies a distance bound and linear confinement, single-shot error correction against random errors, and self correction. While the rest of this section is devoted to proving the above statement, we discuss these consequences properties in \cref{sec:bottleneck}.
\end{remark}


\begin{cor} [Tanner graphs of Diffusion codes on cycle graphs are smaller set expanders]
    Diffusion codes constructed from the interchange process on the cycle graph $\mathcal{C}_N$ have, as Tanner graphs, smaller set expanders.
\end{cor}

\begin{proof}
    This follows from \cref{def:expander codes} and proof of \cref{theorem:diffusioncodecycle}
\end{proof}

The proof of the theorem is presented in the next sub section. The idea is to reduce the smaller set expansion statement to a statement about the distance of particles in the simple exclusion process. In particular, we will need to derive guarantees about \emph{small} numbers of particles $k \ll N$ separating into large pairwise distances already at times $T \ll t_{{\rm mix}, N}$.

In particular, the object of central interest to us will be what we call the \emph{gap vector}.



\begin{defn} [Gap Vector] \label{def:gapvector}
    For $k$ particles on $\mathcal{C}_N$, the gap vector $\mathbf{g}$ is defined as the vector which stores the distance between successive particles. 
    We designate an arbitrary particle as the first and label particles clockwise starting from the first. Then, $g_i$ is the distance between the particle $i$ and particle $i+1$.

\end{defn}

\begin{remark}
The gap vector is subject to the constraint $\sum_i g_i = N$. 
\end{remark}

\begin{remark} 
The mapping from particle positions to gap vectors is many-to-one. However, different particle positions that map to the same gap vector are related by translations. 
\end{remark}



The gap vector allows for facile organization of the states of the SEP. 

\begin{defn} [Partial Ordering] \label{def:partialorderinggap}
    Given two vectors $\mathbf{g}, \mathbf{g}'\in \mathbb{N}^k$, we say that $\mathbf{g}' \preceq \mathbf{g}$ iff $g'_i \leq g_i \forall i$. 
\end{defn}


\begin{remark}
    It is possible that under a cyclic shift of all of the elements of $\mathbf{g}'$, $\mathbf{g}'$ becomes ordered with respect to $\mathbf{g}$. This is accomplished by redefining the particle 1 for the $\mathbf{g}'$ system. 
\end{remark}

\begin{remark}
    If $\mathbf{g}$ and $\mathbf{g}'$ both correspond to states of the SEP on $\mathcal{C}_N$, then $N = \sum_i g_i = \sum_i g_i'$ and it is of course impossible that  $\mathbf{g}' \preceq \mathbf{g}$ unless $\mathbf{g}' = \mathbf{g}$. Therefore, the natural case to have in mind is that $\mathbf{g}'$ corresponds to a system of $k$ particles on a smaller graph $\mathcal{C}_{N'}$ with $N' \leq N$ (see \cref{fig:sep ring}).
\end{remark}

\begin{example} [Smaller gap vector] \label{example:smaller gap vector}
    We can obtain a $\mathbf{g}'$ from $\mathbf{g}$ with $\mathbf{g}' \preceq \mathbf{g}$ via the following. Let $\mathbf{g}$ be the gap vector for a particular particle configuration on $\mathcal{C}_N$. Now, remove a subset of vertices from $\mathcal{C}_N$ which are unoccupied by particles and join up the remaining vertices. That is, if $v_i, v_{i+1},v_{i+2}$ are vertices in $\mathcal{C}_N$ and we remove $v_{i+1}$, now there is an edge $(v_i, v_{i+2})$ in the new graph. Each time a vertex is removed, if that vertex is in the gap $g_i$, then $g_i$ will decrease by 1. This vertex removal process will generate a new $\mathbf{g}'$ that is strictly ordered with respect to $\mathbf{g}$.
\end{example}

We will now define a Markov chain acting directly on $\mathbf{g}$. We will show later that a lazy version of this Markov chain is exactly the evolution of the gap vector in the SEP.

\begin{defn} [Gap Process] \label{def:gapprocess}
    We define the gap process as the following Markov process acting on $\vec g \in \mathbb{N}^k$ (with $g_i > 0\forall i$). 
    Now, each step first choose an entry $i\in [k]$ uniformly at random and then with probability one half each attempt either
        \begin{flalign}\label{eq:gap_increment}
            &&
            g_i(t+1) = g_i(t) + 1,
            &&
            g_{i-1}(t+1) = g_{i-1}(t) - 1,
            &&
        \end{flalign}
        or
        \begin{flalign}\label{eq:gap_decrement}
            &&
            g_i(t+1) = g_i(t) - 1,
            &&
            g_{i-1}(t+1) = g_{i-1}(t) + 1.
            &&
        \end{flalign}
       The updates are rejected if $g_{i-1} = 1$ or $g_{i} = 1$, respectively, to ensure that $g_i > 0\,\forall i$. For the continuous-time version of this chains, let a ``step'' occur upon rings of a Poisson clock of a specified rate. 
\end{defn}

In words, in the gap process at each step we choose a gap at random and then throw a coin to either try and increment or decrement it by 1. A move is rejected if it would shrink a gap to be below 1. 

\begin{remark}
    The above process has an obvious coupling between two copies acting on different initial states: at each time step, we choose the same entry $i$ and the choice of whether to try and increment/decrement. This means we propose the same move on both states, and the state determines whether the move is rejected or not.
\end{remark}

We may also define a "lazy" gap process.

\begin{defn} [Lazy Gap Process] \label{def:lazy gap process}
    We define the lazy gap process as the following Markov process acting on $\mathbf{g} \in \mathbb{N}^k$ (with $g_i > 0\ \forall i$). At each step, with probability $q$ implement an update according to the gap process (\cref{def:gapprocess}), and with probability $1-q$, do nothing. For the continuous-time version of this chains, let a ``step'' occur upon rings of a Poisson clock of a specified rate.
\end{defn}

We will now argue that the SEP induces a lazy gap process. 

\begin{lem} [SEP induces a lazy gap process] \label{lemma:SEP induces gap}
    Let $\mathbf{g}$ be the gap vector corresponding to a state of $k$ particles on the cycle graph $\mathcal{C}_N$, defined according to  \cref{def:gapvector}. Then, the lazy gap process defined in  \cref{def:lazy gap process} with $q = 2k /N$ and clock rate $N$ is the Markov chain induced on the gap vector $\mathbf{g}$ corresponding to the state in the SEP (\cref{def:sep}) on the cycle graph $\mathcal{C}_N$.
\end{lem}

\begin{proof}
We show the result by explicit calculation.
Let $\mathbf{g}$ be a generic $k$ particle gap vector. Let us delineate all the moves of the lazy gap process acting on $\mathbf{g}$ and their probabilities:
\begin{enumerate}
    \item $\forall i$, with probability $p_i^\uparrow = \frac{q}{2k}\mathbb{1}_{g_{i-1}(t) > 1}$ do
    \begin{flalign}
        &&
        g_i(t+1) = g_i(t) + 1
        &&
        g_{i-1}(t+1) = g_{i-1} - 1,
        &&
    \end{flalign}
    \item $\forall i$, with probability $p_i^\downarrow = \frac{q}{2k} \mathbb{1}_{g_{i}(t) > 1}$,
    \begin{flalign}
        &&
        g_i(t+1) = g_i(t) - 1
        &&
        g_{i-1}(t+1) = g_{i-1} + 1
        &&
    \end{flalign}
    
    \item Nothing with probability $p_0 = (1-q) + q \frac{\sum_i\mathbb{1}_{g_i(t) \leq 1}}{k}$.
\end{enumerate}

Next, note that rather than define the SEP as evolving any time a rate 1 clock on an edge rings, we could have equivalently defined the SEP in the following way: Each time a rate $N$ Poisson clock rings, choose an edge uniformly at random to apply a move. 

Our statement follows if we can find $q$ such that these probabilities coincide with the dynamics that is induced on the gap vector in the SEP. We hence have to compute the dynamics of the gap vector in the SEP. 
This is relatively straightforward. Each of the moves in the gap process corresponds to choosing a specific edge in the SEP, so clearly the corresponding probability must be either zero (if the move is impossible) or $1/N$, so 
\begin{flalign}
    &&
    p_i^{\uparrow\prime} = \frac{1}{N}\mathbb{1}_{g'_{i-1}(t) > 1} 
    &&
    p_i^{\downarrow\prime} = \frac{1}{N}\mathbb{1}_{g'_{i}(t) > 1} 
    &&
\end{flalign}
On the other hand, the probability of doing nothing corresponds to choosing an edge that is not adjacent to any particle, which is given by
\begin{equation}
    p_0' = \frac{N - 2 k +2\sum_i \mathbb{1}(g_i'=1)}{N}.
\end{equation}
where $2k-2\sum_i \mathbb{1}(g_i'=1)$ is the number of distinct edges adjacent to a particle.
By observing that the transition probabilities in the lazy gap process and in the SEP agree if $q = 2k/N$, the proof is concluded.

\end{proof}
\begin{cor} [SEP induces a gap process] \label{corollary:sep induces nonlazy}
    The SEP also induces a non-lazy gap process on $\mathbf{g}$ where the clock is of rate $2k$.
\end{cor}

\begin{proof}
    By lemma \cref{lemma:SEP induces gap}, the SEP induces a lazy gap process. The laziness may be removed by choosing a slower clock rate. In this case, the rate $N$ clock becomes a rate $qN = 2k$ clock.
\end{proof}
\begin{remark} \label{remark:particle configs to gaps}
Although the map from particle configurations to gap vectors is inherently many to one, if one has an initial configuration of particles, and the time evolution of the corresponding gap vector $\{\mathbf{g}(1), \mathbf{g}(2), \dots\}$, then the time evolution of that particular configuration of particles is specified. Knowing which gaps shrunk or enlarged at each time step yields which particle moved at each time step.
\end{remark}

The central property that we will use is that the gap proccess has the property of stochastic monotonicity. Intuitively, this says that if we couple to gap processes and start one of the copies in a state with smaller gaps, this copy will remain to have smaller gaps for all times.

\begin{figure}
    \centering
    \includegraphics[width=\linewidth]{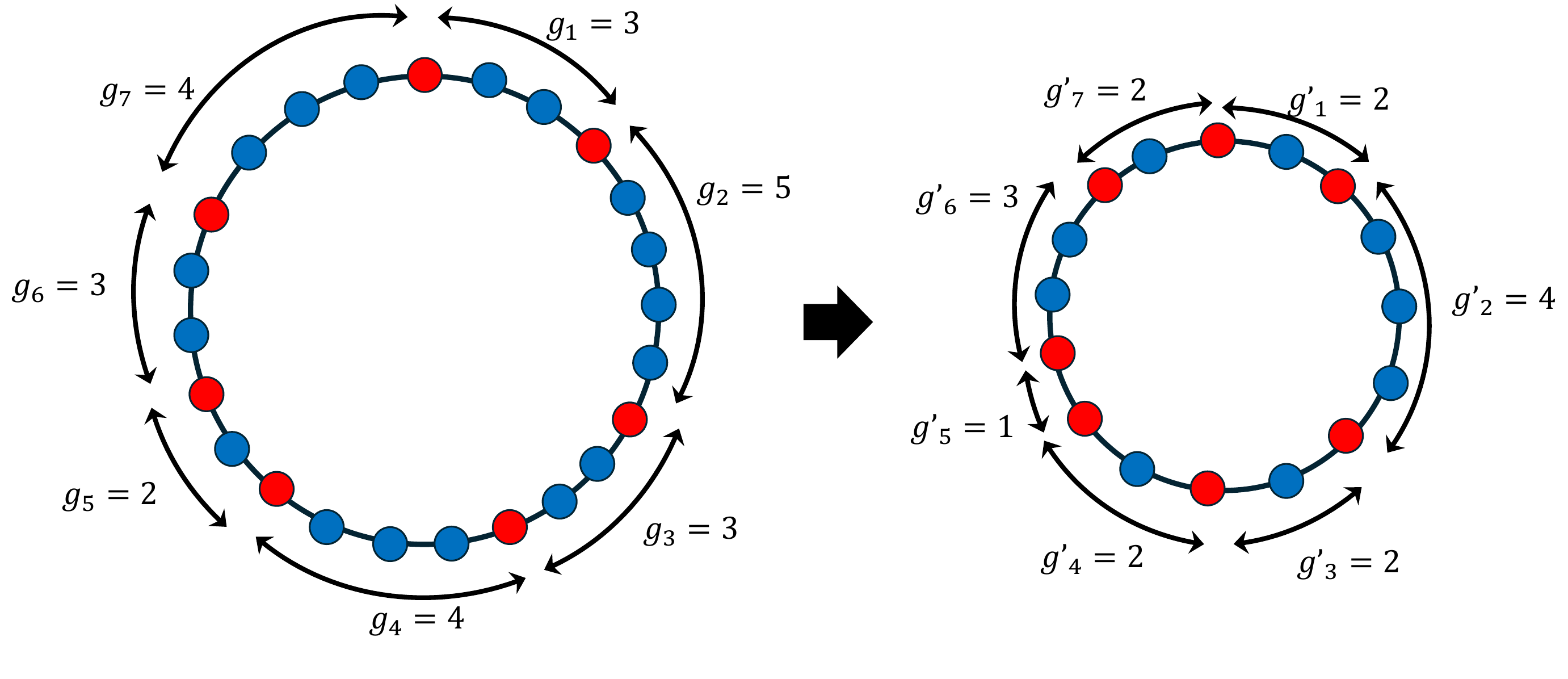}
    \caption{\textbf{Left}: A configuration of $7$ particles on $\mathcal{C}_{24}$. The 12 o'clock vertex is set as the origin and we designate the particle at that vertex as particle 1, we obtain a gap vector where each of the gaps are shown along the ring. The groups of vertices corresponding to a check are drawn. \textbf{Right}: By removing some number of vertices, we obtain a particle configuration on a smaller graph $\mathcal{C}_{16}$. Each of the gaps $g_i'$ are strictly smaller than their counter parts $g_i$.  }
    \label{fig:sep ring}
\end{figure}

\begin{lem}[Stochastic Monotonicity in the Gap Process] \label{lemma:stochasticdomination}
    Let $\mathbf{g}, \mathbf{g}' \in \mathbb{N}$ such that $\mathbf{g}' \preceq \mathbf{g}$ under the partial ordering defined in  \cref{def:partialorderinggap}. Under the obvious coupling of two gap processes, we have $\mathbf{g}'(t) \preceq \mathbf{g}(t)$ for all $t$.
\end{lem}


\begin{proof}
We prove this by induction. The base case is true by assumption.
We then want to show that for all $t$, if $g_i'(t) \leq g_i(t)$, then $g_i'(t+1) \leq g_i(t+1)$.

At time $t$, under the coupling we chose to update the same entry $i$, and whether to try and increment of decrement it, or keep it the same. Clearly, if the update is to not change the entry, then the partial ordering is preserved.
If a move is accepted or rejected on both chains, then the partial ordering is preserved. Now, note that if a move is rejected on $\mathbf{g}$, it is surely rejected on $\mathbf{g'}$, because if $g_i(t)=1$ then by assumption $g_i'(t)=1$. So only case we have to check is a move accepted on $\mathbf{g}$ and rejected on $\mathbf{g'}$.
In this case, we decrease some $g_j$ but do not decrease $g_j'$. However, by assumption $g_j'(t) = 1$ (otherwise the move would not have been rejected on $\mathbf{g'}$) and $g_j(t) > 1$ (otherwise the move would not have been accepted on $\mathbf{g}$) which implies $g_j(t+1) \geq 1 = g_j'(t+1)$.
This concludes the induction step and the statement follows.

\end{proof}
\begin{remark}
    It is important to note here that every $\Delta t = 1$ corresponds to a \textit{single move}, or single ring of the rate $N$ clock. This translates to a ring of one of the $N$ clocks on $\mathcal{C}_N$, when thinking about the gap process in relation to the SEP. 
\end{remark}

\begin{remark}
    By lemma \cref{lemma:stochasticdomination}, we see that if $\mathbf{g}'(0) \preceq \mathbf{g}(0)$, then $\mathbf{g}$ stochastically dominates $\mathbf{g}'$.
\end{remark}

\begin{cor} \label{corollary:monotonic nondecreasing}
    Let $f : \mathbb{N}^k \to \mathbb{R}$ be a monotonically non-increasing function. If $\mathbf{g}',\mathbf{g}$ are two $k$-particle gap vectors evolving under the gap process (\cref{def:gapprocess}) satisfying $\mathbf{g}'(0)\preceq \mathbf{g}(0)$, then $f(\mathbf{g}'(T)) \geq  f(\mathbf{g}(T))$ almost surely 

    This furthermore implies $\mathbb{P}\left[f(\mathbf{g}(T)) \leq d \right] \geq \mathbb{P}\left[f(\mathbf{g}'(T)) \leq d \right]$ for any $d$.
\end{cor}

\begin{proof}
    By definition, $f(\mathbf{g}'(0)) \geq f(\mathbf{g}(0))$ due to the monotonic non-increasing property. By lemma \cref{lemma:stochasticdomination}, we know that $\mathbf{g}'(T) \preceq \mathbf{g}(T)$ almost surely
    . This implies that $f(\mathbf{g}'(T)) \geq f(\mathbf{g}(T))$ almost surely.
\end{proof}

When proving \cref{theorem:diffusioncodecycle}, we will be interested in the probability of there being some number of small gaps in $\mathbf{g}$ at sufficiently large time. Through the lemma \cref{lemma:stochasticdomination}, we have just shown that if $\mathbf{g}'(0) \preceq \mathbf{g}(0)$, then at a later time $T$, $\mathbf{g}'(T) \preceq \mathbf{g}(T)$ almost surely. Furthermore, the number of small gaps $\sum_i \mathbb{1}_{g_i \leq d}$ is a monotonically non-increasing function in $\mathbf{g}$ for a given $d$. Therefore, by corollary \cref{corollary:monotonic nondecreasing}, this implies that the probability of having more small gaps in $\mathbf{g}'$ is almost surely
larger. 

Even given $\mathbf{g}(0)$, it is difficult to obtain general forms forms for the probability distribution of $\sum_i \mathbb{1}_{g_i \leq d}$. However, a useful limit is obtained in the case where $T$ is large enough to exceed the mixing time of the SEP corresponding to $\mathbf{g'}$.

\begin{lem} [Mixing Time of SEP on $\mathcal{C}_N$; theorem 1.1 in \cite{Oliveira_2013}] \label{lemma:SEP mixing time}
    Consider the continuous time SEP on the cycle graph $\mathcal{C}_N$. Then, $\forall \epsilon \in (0,1/2)$, there exists a constant $C > 0$ such that
    \begin{align} \label{eq:sep mixing time}
        t_{\rm mix}(\epsilon) \leq C \cdot N^2 \cdot \ln \left( \frac{N}{\epsilon} \right).
    \end{align}
   
\end{lem}

\begin{remark} \label{remark:clocks and time}
    Note that for a unit of time $\Delta T = 1$, due to the continuous time setting on average $N$ clocks have gone off.
\end{remark}

\begin{cor} [Mixing time of the lazy gap process] \label{corollary:gap mixing time}
  Consider the lazy gap process for a $k$ particle gap vector $\mathbf{g}$ on $\mathcal{C}_N$. Then, $\forall \epsilon \in (0,1/2)$, there exists a constant $C > 0$ such that
    \begin{align} \label{eq:sep mixing time}
        t_{\rm mix}(\epsilon) \leq C \cdot N^2 \cdot \ln \left( \frac{N}{\epsilon} \right).
    \end{align}
   
\end{cor}
\begin{proof}
    By \cref{lemma:SEP induces gap}, the SEP induces a  lazy gap process. Thus, if after a time $t_{\rm mix}(\epsilon) = O(N^2 \log N)$ the state of the SEP is described by the stationary distribution, then the state of the gaps must also be determined by the stationary distribution. Thus $t_{\rm mix}(\epsilon)$ is also the mixing time of the lazy gap process. 

\end{proof}
To use the mixing time bound to bound the number of small gaps after some time $T$, we also need to bound the probability of having a large number of small gaps in the steady state, which is just the uniform distribution over all gap vectors under the constraint $\sum_i g_i = N$. 

\begin{lem} [Number of Small Gaps in Uniform Distribution] \label{lemma:small gaps}

In the uniform distribution over all $k$ particle gap vectors with $\sum_i g_i = N$, 
    
    \begin{align}
        \mathbb{P}\left[\sum_i \mathbb{1}_{g_i \leq d} \geq Q \right] \leq {k \choose Q} \left( \frac{e^2 kd}{N}\right)^Q
    \end{align}
\end{lem}

\begin{proof}
    Under the uniform distribution, every configuration of the gap vector $\mathbf{g}$ is equally likely. How many such configurations are there? This is given by the number of solutions to the equation
    \begin{align}
        g_1 + g_2 + \dots + g_k = N,
    \end{align}
    subject to the constraint that $\forall i, g_i \geq 1$. 

    We can think of the gap vector as a particular distribution of $N$ indistinguishable balls into $k$ indexed bins, with the constraint that each bin must contain at least 1 ball. The number of such distributions is 
    \begin{equation}\label{eq:small_gaps_1}
        N_{\rm gap}(N, k) =
        \begin{pmatrix}
        N-1\\k-1
        \end{pmatrix}.
    \end{equation}

    To proceed, we want to upper bound the number of configurations where at least some number $Q$ of distances is smaller than $d$.
    Again thinking of the gap vector as $N$ balls distributed over $k$ bins, we designate $Q$ bins as special, and designate to each of them a fixed number of $\{\lambda_i\}_{i=1}^{Q}$ balls. Summing over all possibilities for $\lambda_i \leq d \forall i$, and taking a union bound over all choices of special bins we obtain
    
    \begin{align}
        \mathbb{P}\left[ \text{at least $Q$  small bins} \right] &\leq \begin{pmatrix}
            k \\ Q
        \end{pmatrix} \frac{\sum_{\lambda_1 = 1}^d \sum_{\lambda_2 = 1}^d \dots \sum_{\lambda_Q = 1}^d \begin{pmatrix}
            N - \sum_i \lambda_i - 1 \\ (k-Q)- 1
        \end{pmatrix}}{\begin{pmatrix}
        N-1\\k-1
    \end{pmatrix}} 
    \end{align}

    Let us analyze this expression more closely. Because the combinatorics in the numerator only depend on $\sum_i \lambda_i$, we can rewrite the expression in terms of the multiplicity of the vector $\vec \lambda$ under the constraints that $\sum \lambda_i$ and $\lambda_i \leq d\,\forall i$. Dropping the second constraint, we obtain an upper bound for this number, which takes the same form as \cref{eq:small_gaps_1}. 
    Substituting this above yields
    \begin{align}
    &\leq \begin{pmatrix}
            k \\ Q
        \end{pmatrix} \frac{\sum_{s = Q}^{Qd}  
        \begin{pmatrix}
            s - 1 \\ Q - 1
        \end{pmatrix} \begin{pmatrix}
            N - s - 1 \\ k-Q-1
        \end{pmatrix}
        }{\begin{pmatrix}
            N-1\\k-1
        \end{pmatrix}} \\
        &\leq \begin{pmatrix}
            k \\ Q
        \end{pmatrix} \frac{\begin{pmatrix}
            N - 1 - Q \\ k-1 -Q
        \end{pmatrix} \sum_{s = Q}^{Qd}  
        \begin{pmatrix}
            s - 1 \\ Q - 1
        \end{pmatrix} 
        }{\begin{pmatrix}
            N-1\\k-1
        \end{pmatrix}} \\
        &= \begin{pmatrix}
            k \\ Q
        \end{pmatrix} \frac{\begin{pmatrix}
            N - 1 - Q \\ k-1 -Q
        \end{pmatrix}   
        \begin{pmatrix}
            Qd \\ Q
        \end{pmatrix} 
        }{\begin{pmatrix}
            N-1\\k-1
        \end{pmatrix}}
\end{align}
In the second line we upper bounded the second factor in the sum, and in going to the third line we use the hockey-stick identity. 
We can simplify the bound of binomial coefficients
\begin{align}
    \frac{{N-Q-1 \choose k-Q-1}}{{N-1\choose k-1}} &= \frac{{(N-Q-1)!(k-1)!}}{(k-Q-1)!(N-1)!}\\
    &= \frac{(k-1)(k-2)\dots(k-Q)}{(N-1)(N-2)\dots(N-Q)} \\
    &\leq \frac{k(k-1)\dots(k-Q+1)}{N(N-1)\dots(N-Q+1)} \\
    &= \frac{\begin{pmatrix}
        k \\ Q
    \end{pmatrix}}{\begin{pmatrix}
        N \\ Q
    \end{pmatrix}}\\
    &\leq \left(\frac{ek}{N}\right)^Q.
\end{align}
In the third line, we use the fact that $k/N  > (k-Q) / (N-Q)$ and in the last line we use the standard bound on binomial coefficients. Using the same binomial bound on ${Qd\choose Q}$, we finally obtain 
\begin{align}
    \mathbb{P}\left[ \text{at least $Q$  small bins} \right]
    \leq \begin{pmatrix}
        k \\ Q
    \end{pmatrix}
    \left(\frac{e^2 kd}{N} \right)^Q.
\end{align}

\end{proof}

Putting together stochastic monotonicity of the gap process, the mixing time bound on the SEP and its reduction to the gap process, and the bound on the number of small gaps in the steady state of the gap process, we are now finally able to prove the absence of small gaps in the SEP at sufficiently large time.

\begin{lem} [Small Gaps in the SEP] \label{lemma:sep small gaps}
Consider the SEP on $\mathcal{C}_N$ (\cref{def:sep}). Then, there exists a constant $C > 0$ independent of $N$, such that for all $T > 0$, and any $\nu > 0$,
    \begin{align}\label{eq:sep_small_gaps}
        \mathbb{P}\left[ \sum_i \mathbb{1}_{g_i\leq d} \geq Q \right] \leq {k \choose Q} \left( \frac{Ckd}{T^{\frac{1}{2+\nu}}}\right)^Q
    \end{align}
    where $g_i$, $i \in [k]$ are the successive distances of the $k$ particles on $\mathcal C_N$.
\end{lem}

\begin{proof}
    The idea of the proof is to relate the probability on the left hand side of \cref{eq:sep_small_gaps} to the corresponding distribution in the steady state of a gap process on a smaller cycle graph of size $N' < N$ and $T \sim (N')^{2+\nu}$ for any $\nu > 0$.
    
    Consider the SEP on $\mathcal C_N$. The corresponding gap vector $\mathbf{g}$ evolves under a lazy gap process (again, see \cref{lemma:SEP induces gap}) with forward step probability $q = k / N$. Assume w.l.o.g. $k < N'$ (otherwise the bound in \cref{eq:sep_small_gaps} is trivial), then we can find $\mathbf{g'} \preceq \mathbf{g}$ with $\sum_i g_i' = N'$.
    By stochastic monotonicity as shown in \cref{lemma:stochasticdomination}, we then have that under two coupled evolutions, $\mathbf{g'}(T) \preceq \mathbf{g}(T)$. 
    However, due to the coupled evolution, $\mathbf{g}'$ is evolving under a rate $N$ Poisson clock, rather than a rate $N'$ clock. The rate of the clock may be rescaled to $N'$ by rescaling the time on lazy gap process of $\mathbf{g}'$ to $T' = \frac{N}{N'}T$. Thus by \cref{corollary:gap mixing time} if $T' \geq N'^{(2+\omega)}$ for any $\omega > 0$, then the coupled process on $\mathbf{g}'$ is fully mixed. This condition is satisfied if $N^{\prime(2+\omega)}\sim T^{\frac{1}{1+\mu}} \Rightarrow N^{\prime}\sim T^{\frac{1}{2+\nu}}$ for any $\nu > 0$, and hence we can use \cref{lemma:small gaps} to conclude that 
    \begin{align}
        \mathbb{P}\left[\sum_i \mathbb{1}_{g_i \leq d} \geq Q\right] 
        \leq 
        \mathbb{P}\left[\sum_i \mathbb{1}_{g_i' \leq d} \geq Q \right] 
        \leq 
        {k \choose Q} \left( \frac{e^2 kd}{N'}\right)^Q
        \leq
        {k \choose Q} \left(\frac{C kd}{T^{\frac{1}{2+\nu}}}\right)^Q
    \end{align}
    for some constant $C > 0$ and any $\nu > 0$. The first inequality here uses stochastic monotonicity of the gap vector (\cref{lemma:stochasticdomination}), 
    the second is the steady-state distribution bound in \cref{lemma:small gaps}, and the last inequality uses $N^{\prime}\sim T^{\frac{1}{2+\nu}}$.

\end{proof}

\subsubsection{Proof of \cref{theorem:diffusioncodecycle}}

We are now finally able to proof our main result. We consider diffusion codes created by a swap process on the cycle graph $\mathcal C_N$ and show that they are smaller set lossless expanders.

\begin{proof}

The proof will proceed through several arguments. 
 
   \begin{enumerate}
       \item Recall that each vertex of this graph corresponds to an \emph{edge} in the Tanner graph $\mB = (L,R,E)$, and bits and checks correspond to groups of such vertices of size $\wbit$ and $\wcheck$, respectively. Our goal will be to show that for Tanner graphs generated from the process defined in  \cref{def:diffusioncodes}, the probability that there exists a subset of bits $S \subset L : |S| \leq \delta (n)$ with neighbor set size $|\Gamma(S)| \leq (1-\varepsilon ) \wbit |S|$ isbounded away from 1. We denote this probability by $\mathbb{P}\left[|\Gamma(S)| \leq (1-\varepsilon ) \wbit |S|)\right]$.
       \item Denoting the set of outgoing edges from $S$ by $F$, the number of neighbors of $S$ is simply equal to the number of unique checks that $F$ connects into, the set of which we denote by $\Sigma(F)$.
    Hence, to show the above it suffices to show that for all $F \subset E$ such that $\abs{F} \leq \wbit\,\delta(n)$, $\mathbb{P}\left[|\Sigma(F)| \leq (1-\varepsilon) \abs{F})\right]$ is bounded away from one.

    A subset of edges $F$ corresponds to a subset of vertices $U$ in the cycle graph, and the number of unique checks corresponds to the number of unique check groups, $\Xi(U)$, that $U$ is connected to after the swap process. Again, to show the claim it hence suffices to show that after the swap process, for $\abs{U} \leq \wbit  \delta(n)$, $\mathbb{P}\left[|\Xi(U)| \leq (1-\varepsilon) \abs{U})\right]$ is bounded away from 1.
    
    \item     Now, given a subset of vertices $U$, their evolution under the interchange process corresponds to an instance of the simple exclusion process on $\mathcal C_N$ with the initial condition that particles are placed on $U$. Let $\mathbf{g} \in \mathbb{N}^{\abs{U}}$ be the vector storing all successive interparticle distances, the gap vector (\cref{def:gapvector}).  The number of unique check groups $\abs{\Xi(U)}$ that $U$ is connected to is then given by the number of unique groups that these particles fall into after the evolution under the SEP. Note that two particles cannot be in the same group if their distance is larger than $\wcheck$.
    Because of this, 
    $\mathbb{P}\left[|\Xi(U)| \leq (1-\varepsilon )\abs{U})\right] \leq \mathbb{P} \left[ \sum_i \mathbb{1}_{g_i \leq \wcheck} \geq \varepsilon \abs{U} \right].$
    
    \item We have reduced the problem to determining the probability distribution of the number of small gaps in the SEP. 
    This is an SEP of $|U| = \wbit |S|$ particles on the cycle graph of size $N$. Then, by \cref{lemma:sep small gaps}, 
\begin{align}
    \mathbb{P} \left[ \sum_i \mathbb{1}_{g_i \leq \wcheck} \geq \varepsilon \wbit |S| \right] \leq \begin{pmatrix}
        \wbit |S| \\ \varepsilon \wbit |S|
    \end{pmatrix} \left(\frac{e^2 \wcheck \wbit |S|}{T^{\frac{1}{2+\nu}}} \right)^{\varepsilon \wbit |S|},
\end{align}
for all $\nu > 0$.

\item Let us now set $\delta(n) = \frac{T^{\frac{1}{2+\nu}}}{e^2 \wcheck}$. Then, by the line of argument above, since the LHS of the above inequality also upper bounds $\mathbb{P}\left[ |\Gamma(S)| \leq (1-\varepsilon )\wbit |S|\right]$, we obtain
\begin{align}
    \mathbb{P}\left[ |\Gamma(S)| \leq (1-\varepsilon )\wbit |S|\right] &\leq \begin{pmatrix}
        \wbit |S| \\ \varepsilon \wbit |S|
    \end{pmatrix} \left(\frac{\wbit |S|}{\delta(n)} \right)^{\varepsilon \wbit |S|} \\
    &\leq \left(\frac{e\wbit |S|}{\varepsilon \delta(n)} \right)^{\varepsilon \wbit |S|}.
\end{align}

This bound is of the exact same form as \cref{eq:sublinear bound}, and so the remainder of the proof follows identically to that of \cref{theorem:sublinearlossless}.

\item As in the proof of \cref{theorem:sublinearlossless}, we demand that $\delta(n) = c_1n^\beta$ for some $\beta < 1$ where $c_1$ is some positive constant. Then, we require that $T = (e^2 \wcheck c_1 n^\beta)^{2+\nu}.$ Since $\nu$ can be made arbitrarily close to 0, if we define $\alpha > \beta/2$ then we have shown that for $T \propto n^\alpha$ with $0 < \alpha < 2$, we can obtain a $\beta < \alpha / 2$ which sets a scale ofexpansion. The resultant diffusion code is with nonzero probability a $(\delta(n), \gamma)$ smaller set expander.

    \item Also, as in the proof of \cref{theorem:sublinearlossless}, for a chosen $\varepsilon$, the expansion scale must satisfy $\beta > (\varepsilon c)^{-1}$, where $c$ was the left degree in the construction of \cref{theorem:sublinearlossless}. As the proof here follows identically, once again, as long as $\beta > (\varepsilon \wbit)^{-1}$, then the desired $\varepsilon$ is achieved with nonzero probability. This also implies that $\alpha / 2 > (\varepsilon \wbit)^{-1}$ is sufficient. 

   \end{enumerate}

\end{proof}

We conclude this section by discussing \textit{right side} lossless expansion in diffusion codes, as well as bounds on the size of checks.

\begin{lem} [Lossless right vertex expansion of diffusion codes on cycle graphs] \label{lemma:diffcode rightside}
    Take $\mathcal{C}_N$ the cycle graph of size $N = n \wbit = m \wcheck$, and its obvious partition into contiguous groups of $\wcheck$.  
        For all $\varepsilon, \alpha, \beta > 0$, if $T \propto n^\alpha$ and $(\varepsilon \wbit)^{-1} < \beta < \alpha/2$, then the Tanner graph of the $(n, m, \wbit,\wcheck, \mathcal{C}_N, T)$ diffusion code is, with probability $q(N) > 0$
        , a $(\delta, \gamma)$ lossless smaller set \textit{right} vertex expander with $\delta \propto n^{\beta}$ and $\gamma > \wbit(1-\varepsilon)$. Further, $q(N)\to 1$ as $N\to\infty$.
    
\end{lem}

\begin{proof}
    The interchange process used to generate diffusion codes is time reversal symmetric. Therefore, any instance of the process can be viewed as an instance of the same process with bits and checks flipped. Then, the proof follows nearly identically to that of the proof of \cref{theorem:diffusioncodecycle}. 
\end{proof}

\begin{cor} [Left/right expansion of diffusion codes on cycle graphs] \label{corollary:left right expansion diffusion}
    Diffusion codes with $n$ bits and $m$ checks as constructed in \cref{theorem:diffusioncodecycle} are, with nonzero probability, left $(\delta_L(n), \gamma_L)$ and right side $(\delta_R(n), \gamma_R)$ sublinear expanders.
\end{cor}

\begin{proof}
    This directly follows from \cref{theorem:diffusioncodecycle} and \cref{lemma:diffcode rightside}.
\end{proof}


\begin{lem}\label{lemma:check distance}
    Let $\mathcal C_N$ be the cycle graph of size $N=n\wbit=n\wcheck$ and its obvious partition into contiguous groups of $\wbit$ and $\wcheck$. Then, there exists a constant $C > 0$ such that for all $\alpha>0$, the $(n, m, \wbit, \wcheck, \mathcal C_N, T=n^\alpha)$ diffusion code is quasi-local on $\mathcal C_N$ in the sense that for all $\mu>0$ the probability of any check being (geometrically) larger than $C T^{\frac{1}{2}+\mu}$ vanishes as $n\to\infty$.
\end{lem}


\begin{remark}
    In a standard expander code such as a Gallager code, each check is equally likely to connect to any of the bits in the system. Therefore, heuristically, we would expect that for any given check $j$ in a Gallager code, the average edge length of the edges between it and the bits to which it connects would be $\mathcal{O}(n)$. In the case of the diffusion codes there is a sense of locality to the problem. From the definition of diffusion codes (\cref{def:diffusioncodes}), and the initial numbering demanded in \cref{lemma:check distance}, any check $j$ will start off being connected to a bit that is within small $\mathcal{O}(1)$ distance of $j$. As the diffusion process unfolds, each check "spreads out" and becomes connected to farther and farther bits causing an increase in the check to bit distances. Trivially this distance is upperbounded by the total number of swaps performed during the interchange process (i.e. $NT$), but in actuality, the check distances are much smaller.
\end{remark}

\begin{proof}

    For any socket, if we follow its evolution under the interchange process, it is just a simple random walk on $\mathcal{C}_N$. From the continuous time nature of the problem, the probability that the particle steps clockwise is determined by the probability that the Poisson clock on the edge to the right of the particle goes off before the clock to the left of the particle. This is simply $1/2$, since the two clocks are independent exponentials of rate 1.
    
    The check size is then determined by the displacement of a particle from its original socket under a simple random walk process. On $\mathcal{C}_N$, the largest this displacement can be is $N/2$. Had this simple random walk occurred on the infinite line $\mathbb{Z}$, the displacement could only be greater and would thus upper bound the displacement on $\mathcal{C}_N$. We therefore consider instead the simple random walk on $\mathbb{Z}$. Letting $\Delta_T$ be the displacement of the particle at time $T$, by the Azuma-Hoeffding inequality, 
    \begin{align}
        \mathbb{P}\left( \Delta_T \geq \epsilon \right) \leq 2 \exp \left( \frac{-\epsilon^2}{2T}\right).
    \end{align}

    If we set $\epsilon = \sqrt{2} T^{\frac{1+\nu}{2}}$, where $\nu$ is an arbitrarily small positive number, the inequality becomes
    \begin{align}
        \mathbb{P}\left(\Delta_T \geq 2T^{\frac{1+\nu}{2}} \right) \leq 2 \exp \left(-T^\nu\right).
    \end{align}

    Now, if we let $T = n^\alpha$, and $\mu = \nu/2$, then the inequality becomes
    \begin{align}
        \mathbb{P}\left( \Delta_T \geq 2T^{\frac{1}{2}+\mu} \right) \leq 2 \exp \left(- n^{2\alpha\mu} \right)
    \end{align}
    
    Thus, the probability that the particle wanders away from its initial position by more than $N^{\alpha /2}$ vanishes as a stretched exponential. Since the edge we expected was arbitrary, we obtain
    
    \begin{align}
        \mathbb{P}\left(\exists\ \text{edge with length }\geq 2 T^{\frac{1}{2}+\mu} \right) \leq 2\wbit n \exp(-n^{2\alpha\mu}).
    \end{align}
    where the right hand side vanishes as $n\to\infty$.

\end{proof}

\subsection{Consequences of Smaller Set Expansion} \label{sec:bottleneck}

In this section, we discuss consequences of expansion as present in diffusion codes, such as bounds on distance, linear confinement and self correction. We additionally discuss how the smaller set expansion property guarantees properties of the quantum codes constructed from hypergraph products.

\subsubsection{Distance lower bound and confinement}

From the definition of expander codes (\cref{def:expander codes}), diffusion codes are expander codes with their Tanner graphs being $(\delta(n), \gamma)$ expanders, with $\delta(n) = o(n)$. They therefore inherit many of the favorable properties of expander codes, provided that $\gamma$ is large enough. If $\gamma = \wbit (1-\varepsilon)$, then from \cref{lemma:expansion distance}, if $\varepsilon < 1/2$, diffusion codes have distance lower bounded by $\delta(n)$. Furthermore, from \cref{lemma:expansion confinement}, they also exhibit $(\delta(n), \wbit(1-2\varepsilon)$ confinement (\cref{def:confinement}). As is often the case for diffusion codes, $\delta(n)$ is a sublinear power law in $n$, the total number of bits. If $\delta(n) \sim n^\beta$, with $\beta < 1$, then one obtains a distance lower bounded growing as $n^\beta$ as well as confinement out to a scale which is $n^\beta$. This allows one to construct families of diffusion codes with arbitrary distance and confinement cutoff scaling if the bit and check degree is chosen appropriately.

\subsubsection{Self correction} 

Classical linear codes admit a natural interpretation as a classical Hamiltonians of Ising variables (see \cref{eq:hamiltonian_classical}). A code is called self-correcting if this Hamiltonian, when coupled to a thermal bath, passively preserves the code space.
The codewords are the ground states of its Hamiltonian, and physically self correction hence means that under coupling to a bath the system stays ``near'' its ground state for a long time.

A natural mathematical model for this idea is Glauber dynamics: each time step we choose a random bit and flip it with a probability determined by a Boltzmann factor $p_{\rm flip} \propto \exp(-\beta \Delta E)$, where $\Delta E$ is the difference in the size of the syndrome before and after the flip. Then, a code is called self correcting if under this dynamics it takes a number of steps exponentially large in (some power of) the system size to produce an uncorrectable error when starting from a codeword.
%
%
We formalize this idea in the following two definitions:

\begin{defn} [Memory Time] \label{def:memory time}
    Consider a code initialized in a code state evolving under Glauber dynamics. If, for some decoder, there is a $t_{\rm mem}$ at which that decoder fails to recover the code state, this time is deemed the memory time of the code under that decoder.
\end{defn}

\begin{defn} [Self Correction] \label{def:self correction}
    Given a family of codes $\mathsf{C}_i$ of increasing size, we say that the family is self-correcting if the memory time of the codes diverge with the size of the code.
\end{defn}

Remarkably, self-correction is a direct consequence of confinement.

\begin{theo} [Self Correction from Confinement (Theorem S IV.10 in \cite{placke2024topologicalquantumspinglass})] \label{theorem:self correction from confinement}
    Consider a family of expander codes $\{\mathsf{C}_i\}$, $i\in \mathbb{N}$ with $n_i$ bits and $m_i$ checks. If the family has $(\delta(n_i), \gamma)$ confinement for some $\gamma > 0$ and $\delta(n) = \omega (\log (n))$ then at sufficiently low temperature the family $\mathsf{C}_i$ is self-correcting (\cref{def:self correction}) with memory time diverging as $\mathcal{O}(e^{\delta(n)})$.
\end{theo}

\begin{proof}
    The proof follows the proof of theorem S IV.10 in \cite{placke2024topologicalquantumspinglass}. As the smaller set expanders have $\delta(n) = \omega(\log (n))$, the condition outlined in theorem S IV.10 of \cite{placke2024topologicalquantumspinglass} is met and the proof is instantiated.
\end{proof}

Confinement in turn is a direct consequence of (sufficiently strong) vertex expansion of the Tanner graph.

\begin{cor} [Self Correction from Smaller Set Expansion]
    Consider a family of expander codes $\{\mathsf{C}_i\}$, $i\in \mathbb{N}$ with $n_i$ bits and $m_i$ checks. If the Tanner graphs of the family are $(\delta(n_i), \gamma)$ smaller set expanders with $\gamma > \wbit /2$ and $\delta(n) = \Omega(n^\beta)$ for some $0 < \beta < 1$, then the family $\{\mathsf{C}_i\}$ is self-correcting.
\end{cor}

\begin{proof}
    From \cref{lemma:expansion confinement}, smaller set $(\delta(n), \gamma)$ expansion of the Tanner graph implies $(\delta(n), \gamma)$ confinement. The proof then directly follows from the fact that $\delta(n) \propto n^\beta$ is strictly larger than $\log (n)$. Then by \cref{theorem:self correction from confinement}, the corollary is proven.
\end{proof}

In addition to the rigorous results presented here, in \cref{sec:classical self correction}, we numerically explore self correction in diffusion codes. There, we determine lower bounds on the critical temperature needed for self correction and compute the memory times as a function of temperature and system size.

\subsubsection{Linear-time decoders}

Sipser and Spielman's proved the existence of a linear time decoding algorithm called the $\mathsf{flip}$ decoder \cite{sipserspielman} using the fact that the Tanner graphs of their codes exhibit lossless unique neighbor expansion. They show that for sufficiently large constant $\gamma$, the $\mathsf{flip}$ decoder corrects any error up to a finite fraction of the distance. This proof also applies to diffusion codes, which have distance scaling subsextensively with $n$. Nevertheless, we expect the flip decoder to correct \emph{random} errors at sufficiently error rates. This is because as shown e.g. in \cite{bombin2015single_shot,Fawzi_2018_lineartime} at sufficiently low error rates random errors \emph{cluster} into connected components that can be decoded separately, and no single component contains more than $O(\log n)$ bits. We thus expect any local decoder for LDPC codes which corrects more than $\log n$ adversarial errors to have a threshold against random errors as well.
Indeed, in \cref{sec:decoding benchmarks} we present numerical evidence that the flip decoder has a threshold for families of diffusion codes with sufficiently large check degree.



\subsubsection{Confinement in the Hypergraph Product} 

The properties of hypergraph product quantum codes are determined by the underlying classical input codes. 
Most importantly for our purposes, one can guarantee the expansion of the hypergraph product Tanner graphs given (sufficiently strong) expansion of the inputs.
This in turn implies a generalized notion of confinement in the corresponding quantum code. This notion of confinement in turn implies quantum self correction, single-shot error correction, and the existence of a linear-time decoder against random errors.


The reason why the notion of confinement has to be revised for quantum codes is that the size of the syndrome in an qLDPC codes cannot be proportional simply to the size of the error. This is because there are small errors that have exactly zero syndrome but also trivial action on the codespace: the stabilizers. The solution is to define confinement with respect to a modified notion of error size, which we call \emph{reduced weight}. This notion measures the distance of any error from products of stabilizers or, equivalently, the smallest weight of any error with the same action on the code space.

\begin{defn} [Reduced Weight] \label{def:reduced weight}
    Given a pair of parity check matrices $\mathbf{H}_X$ and $\mathbf{H}_Z$ which satisfy $\mathbf{H}_X \cdot \mathbf{H}_Z^T = 0$, the reduced weight with respect to $\mathbf{H}_X$ and $\mathbf{H}_Z$ is defined as
    \begin{align}
        ||\mathbf{x}||_X &:= \text{dist}[\mathbf{x}, \text{im} (\mathbf{H}_X^T)] \\
        ||\mathbf{z}||_Z &:= \text{dist}[\mathbf{z}, \text{im} (\mathbf{H}_Z^T)],
    \end{align}
    where for some $\mathbf{x} \in \mathbb{F}_2^n$ and some subspace $A$, $\text{dist}(\mathbf{x},A) := \min_{\mathbf{a} \in A} |\mathbf{x}+\mathbf{a}|$.
\end{defn}

Given this adjusted notion of weight, we define confinement in quantum codes as before. Quantum CSS codes have two kinds of checks and errors, and correspondingly there are two distinct notions of confinement. We call these \emph{boundary} and \emph{coboundary} confinement, borrowing the name from the language of chain complexes. In this language, the stabilizers of the quantum CSS code are identified with the boundaries and co-boundaries of the complex, respectively, and the reduced weight measures the distance of a chain from the space of boundaries and co-boundaries, respectively.

\begin{defn} [Confinement in Quantum Codes] \label{def:quantum confinement}
Consider a quantum CSS code with parity check matrices $\mathbf{H}_X$ and $\mathbf{H}_Z$. For some $\delta, \gamma > 0$, we say that the code has $(\delta, \gamma)$-boundary confinement if
\begin{align}
    ||\mathbf{x}||_X \leq \delta(n) \Rightarrow |\mathbf{H}_Z \mathbf{x} | \geq \gamma ||\mathbf{x}||.
\end{align}

and we say that the code has $(\delta, \gamma)$-coboundary confinement if
\begin{align}
    ||\mathbf{z}||_Z \leq \delta(n) \Rightarrow |\mathbf{H}_X \mathbf{z} | \geq \gamma ||\mathbf{z}||.
\end{align}

We say that the code has $(\delta, \gamma)$ confinement if it has both boundary and coboundary confinement.
    
\end{defn}


In a hypergraph product quantum CSS code, confinement can be guaranteed provided that the underlying classical codes display both left and right unique neighbor expansion (\cref{def:unique neigbhor expansion}) \cite{Leverrier_2015}. Thus, hypergraph products of codes with sufficiently strong smaller set expansion, such as certain families of diffusion codes, will also exhibit confinement.

\begin{theo} \label{theorem: sublinear expander hypergraph product}
    Consider a code $\mathsf{C}$ with $n$ bits, $m$ checks, max bit degree $\wbit$ and max check degree $\wcheck$. Let $\varepsilon>0$ such that $\wbit, \wcheck > \epsilon^{-1}$, and let the Tanner graph $\mB = (L, R, E)$ of $\mathsf{C}$ be a $(\delta, \wbit (1-\varepsilon))$ left-right vertex expander. Then, the quantum code $\mathsf{C}_Q = \{ \mathsf{C}_X, \mathsf{C}_Z\}$ constructed as the hypergraph product of $\mathsf{C}$ with itself (\cref{def:hypergraph product}) has $(\delta_Q, \gamma_Q)$ confinement (\cref{def:quantum confinement}). Here, $\delta_Q = \alpha\,\delta$ with $\alpha$ a positive constant that depends only on $\wbit, \wcheck$ and $\epsilon$, and $\gamma_Q = \tfrac{1}{2}(1 - 8\varepsilon)\wbit$.
\end{theo}

\begin{proof}
    The proof follows analogously to that of Lemma 15 in \cite{Fawzi_2018}.
\end{proof}

The above means in particular that for any power $\beta > 0$, there exists a family of diffusion codes with $T\sim n^\beta$ such that their hypergraph product yields a family of quantum codes with confinement.

\begin{cor} [Hypergraph product of diffusion codes]
    Consider a family of diffusion codes defined on the cycle graph with $T\sim n^\beta$ for some $\beta > 0$ and $\wbit, \wcheck$ sufficiently large. Then, the family of quantum qLDPC codes obtained as the hypergraph product of the diffusion code with themselves has with high probability $(\delta_Q(n_Q), \gamma)$ confinement (in the sense of \cref{def:quantum confinement}), with $\delta \sim n_Q^{\beta/2}$, where $n_Q$ is the number of qubits in the hypergraph product code.
\end{cor}

\begin{proof}
    From \cref{corollary:left right expansion diffusion}, families of diffusion codes constructed from the cycle graph yield, with high probability, lossless left-right vertex expanders. The corollary then follows from \cref{theorem: sublinear expander hypergraph product}.
\end{proof}

Perhaps the most striking consequence of confinement in quantum LDPC codes is self correction, defined analogously to \cref{def:self correction}. As in the classical counterparts, any family with $(\delta(n), \gamma)$ confinement with super-logarithmic $\delta(n)$ and $\gamma > 0$ is self-correcting \cite{Hong_2025,placke2024topologicalquantumspinglass} with memory time scaling as $t_{\rm mem} = \exp[\Omega(\delta(n))]$. Hence, it is easy to see that hypergraph products of diffusion codes with $T\sim n^\beta$ for any $\beta > 0$ and sufficiently large bit and check degree are qunatum LDPC codes with confinement in the sense of \cref{def:quantum confinement} and hence self-correcting.

A numerical characterization of the self-correcting properties of diffusion codes on the same level as done for the classical codes in \cref{sec:classical self correction} is beyond the scope of this manuscript. However, we do perform a minimal ``heating experiments' on a single instance of a diffusion codes in \cref{sec:quantum self correction}, and find the result to be consistent with a large memory time at low temperatures.

Finally we mention that quantum codes with sufficiently strong enough confinement admit the use of a linear time small-set-flip decoding algorithm decoding algorithm (for $\gamma > 7\wbit / 8$) . If the expansion strength is larger still the code small set flip decoding algorithm is able to correct errors under noisy syndrome measurements \cite{quantumexpandercodes, Fawzi_2018, Fawzi_2018_lineartime} (for $\gamma > 15\wbit / 16)$). 
We thus expect hypergraph product diffusion codes with strong enough expansion to also be endowed with linear time decoding and single shot decoding.

\section{Numerical Experiments}

In this section, we showcase various numerical experiments performed on diffusion codes both in the classical and quantum case.

\subsection{Benchmarks against i.i.d. Noise} \label{sec:decoding benchmarks}

Sipser and Spielman \cite{sipserspielman} showed that expander codes with strong enough expansion (as determined by the parameter $\varepsilon$ in \cref{def:lossless expansion}) admit the use of a local active decoding algorithm, i.e. the $\mathsf{flip}$ decoder. This decoding algorithm works by checking each bit to see if the total number of unsatisfied parity checks may be lowered by flipping it. If so, the algorithm flips the bit and continuous on. 

In particular, they proved that, given a family of $(\delta(n), \gamma)$ expander codes, the flip decoder can correct adversarial errors up to a size $\sim \delta(n)$. It also works for random errors, even though these typically have weight much larger than $\delta(n)$. This is because percolation theory tells us that at sufficiently low error rates (below the percolation threshold), errors form connected clusters which are at most of size $\sim \log (n)$, and these clusters can be decoded independently. Hence, as long as $\delta(n)$ grows faster than $\log(n)$, the flip decoder indeed serves as a viable decoder for these codes under the random error model \cite{Fawzi_2018, Fawzi_2018_lineartime}. 
The diffusion codes we introduced in this work have smaller set expansion with $\delta(n)$ growing as some power law, and hence also allow for the use of the flip decoder provided that $\wbit, \wcheck$ and $T$ are chosen to make $\varepsilon$ sufficiently small. 
To test performance of this procedure and determine the numerical threshold, we perform decoding simulations. The experiment is performed as follows: 1) Begin in the all 0 code state. 2) Flip some fraction $p$ of bits randomly. 3) Run a version of the flip decoder corresponding to zero-temperature Metropolis dynamics. The decoder terminates when one of two conditions is met: 1) Bits are flipped until all checks are satisfied. 2) Bits are flipped until there is no single bit flip that decreases the total number of unsatisfied checks. The remaining error left over in this latter condition is what is termed a stopping set. 

We perform numerical experiments using diffusion codes with $\wbit = 9, \wcheck = 11$ and $T = N$. We expect this to result in, up to constants and log factors, correspond to an expansion scale of $\sim \sqrt{N}$. Note that although in \cref{sec:construction} we defined the diffusion code construction from a continuous time process, when numerically constructing the codes, we implement the SWAP network in discrete time, performing one SWAP per time step. Subsequently, this translates to overall $ NT = N^2$ SWAP gates applied. As shown in the left panel of \cref{fig:threshold}, we see that indeed, these codes have a threshold under the $\mathsf{flip}$ coder. For sufficiently bit-flip rates $p$, the logical failure rate $\mathbb{P}_{\rm fail}$ decays quickly with system size. 
To obtain the plotted data, we average over 10 instances of diffusion codes for each system size $n$ and perform $10^4$ decoding experiments per code. Each point in the figure is hence an average over $10^5$ decoding experiments. 
The threshold is numerically determined to be roughly $0.017 \leq p_c^{\rm flip} \leq 0.019$. Below the threshold, we expect the logical failure rate to decay as a stretched exponential in system size, and as shown in the inset our data is consistent with such behavior.

\begin{figure}
    \centering
    \includegraphics[width=0.49\linewidth]{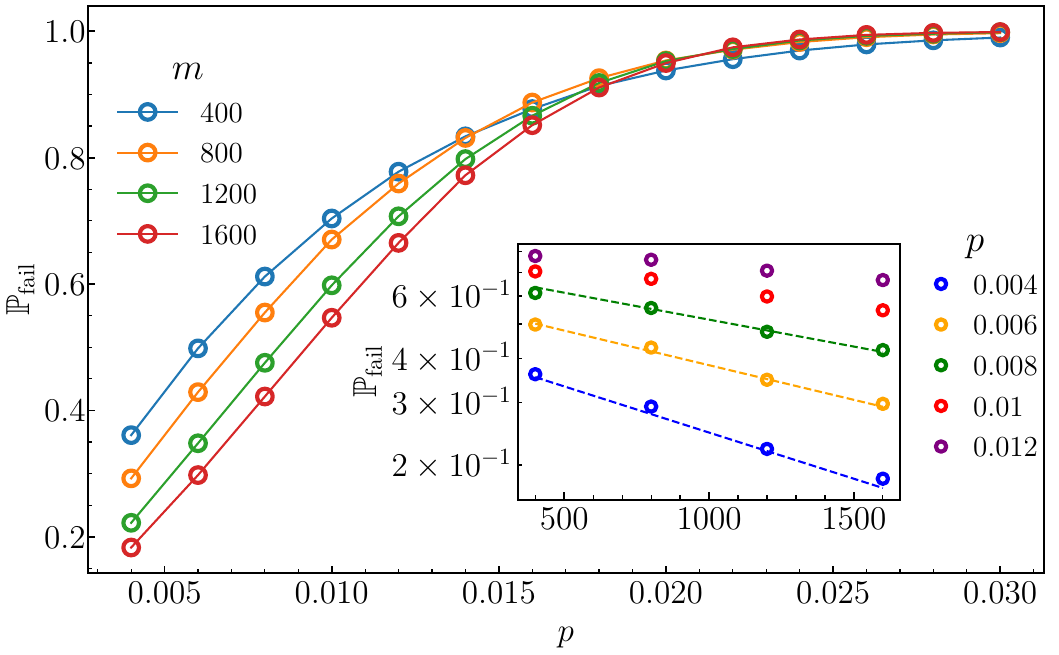}
    \includegraphics[width = 0.49\linewidth]{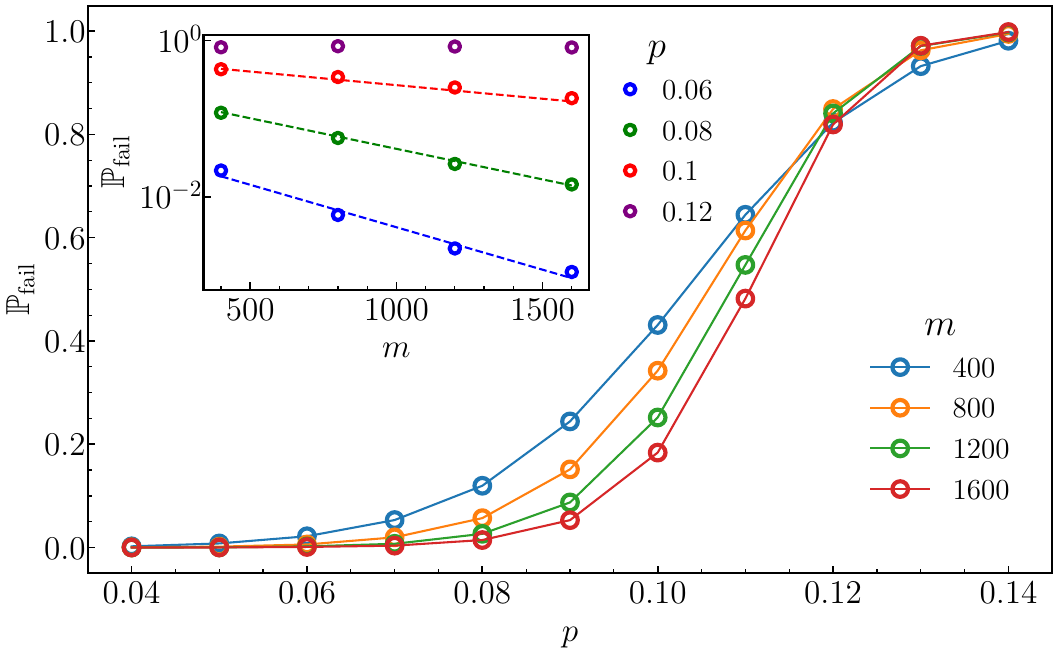}
    \caption{\textbf{Left}: Threshold of $\wbit = 9$, $\wcheck = 11$, $T = N$ diffusion codes under the $\mathsf{flip}$ decoder. Inset shows below threshold, the failure rate decays with system size. \textbf{Right}: Threshold of $\wbit = 9$, $\wcheck = 11$, $T = N$ diffusion codes under the BP decoder.
    Inset shows below threshold, the failure rate decays with system size. In both insets, dashed lines are not fits, but meant to guide the eye. 
    }
    \label{fig:threshold}
\end{figure}
  

We also benchmark our codes using a belief propagation (BP) decoder \cite{Roffe_2020, Roffe_LDPC_Python_tools_2022}, which has been shown to work well on both LDPC and qLDPC codes. This data is shown in the right panel of \cref{fig:threshold}. BP
performs significantly better than the $\mathsf{flip}$ decoder, and we obtain a threshold of roughly $0.11 \leq p_c^{\rm BP} \leq 0.13$.
The data is again averaged over $10^5$ data points, with $10^4$ decoding simulations run on $10$ instances of the code. We use the same 10 instances to produce the data for both the $\mathsf{flip}$ and BP decoder. Similar to the $\mathsf{flip}$ decoder, the decay of the failure rate below the threshold is consistent with a stretched exponential. However, we note that in both cases it is hard to conclusively rule out other functional forms given the limit range of system sizes.

The performance of codes in any decoder is highly dependent on the chosen parameters. In particular, one requires $T \gg \wbit, \wcheck$; otherwise, the code is subject to large finite size effects. If $\wbit, \wcheck$ or $T$ are made to be too small, then the decoder ceases to be viable at all and the code fails spectacularly with almost certainty on each trial. For example, if one desires a code with $\wbit = 5$ and $T = n^{1/4}$, then if $n=100$, then $T$ will actually be smaller than $\wbit$, indicating that the edges of a check have not diffused outwards at all.

\subsection{Glauber Dynamics}

\subsubsection{Self correction in classical diffusion codes} \label{sec:classical self correction}

\begin{figure}
    \centering
    \includegraphics[width=0.49\linewidth]{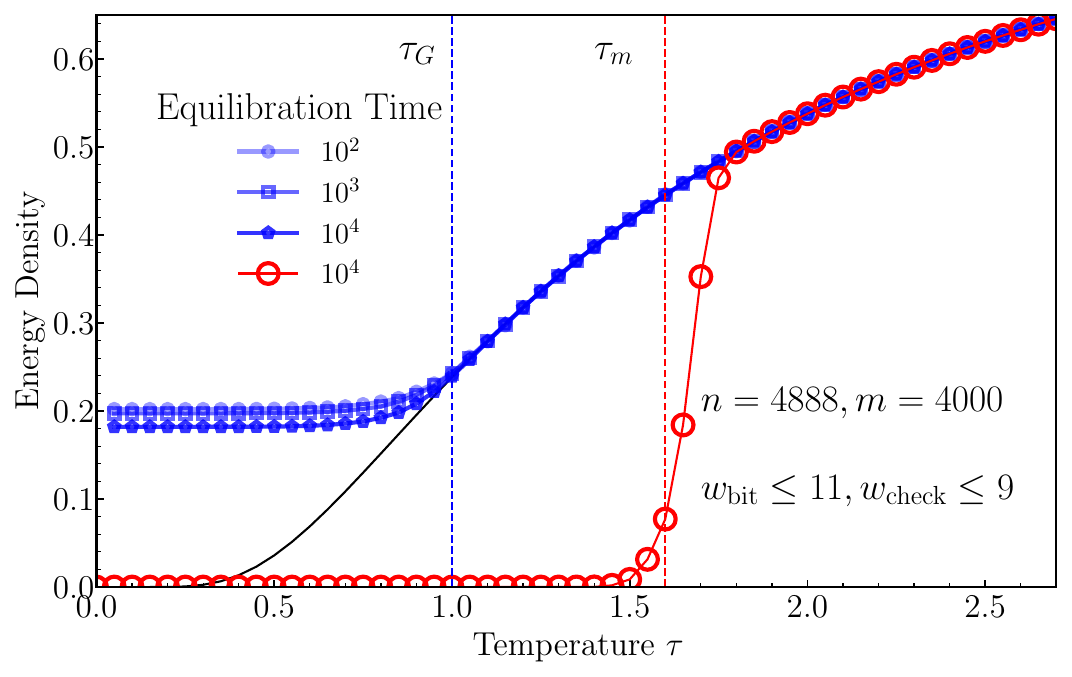}
    \includegraphics[width=0.49\linewidth]{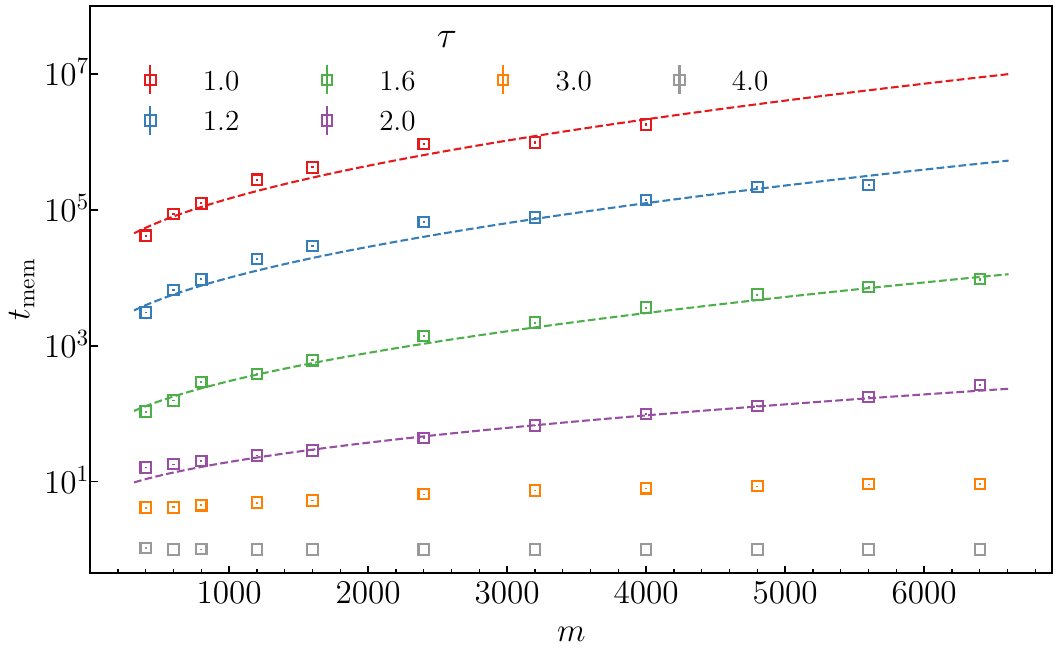}
    \caption{\textbf{Left}: \textbf{Heating and cooling experiments}. In this panel we show heating and cooling experiments done on $\wbit = 9, \wcheck = 11, T = N$ diffusion codes. Black line indicates the predicted thermal equilibrium value. At $\tau_m > 1.5$, we observe an abrupt transition where the code jumps from near 0 energy density to the predicted thermal equilibrium value. When cooling the system back down to $\tau = 0$, the system becomes trapped in a local minimum and falls off equilibrium. 
    \textbf{Right}: \textbf{Memory Time Experiments}. In this panel, we show the results of memory time experiments for $\wbit = 9, \wcheck = 11, T = N$ diffusion codes, using the $\mathsf{flip}$ decoder for read-out. For $\tau < 2$, we observe an increase in the memory time with system size that is consistent with a stretched exponential as predicted by \cref{theorem:self correction from confinement}.
    }
    \label{fig:selfcorrection}
\end{figure}

We also test the performance of the codes against \textit{thermal noise} in order to demonstrate their self-correction ability. We do this by first constructing a Hamiltonian composed of the parity checks. For a code of $n$ bits, denoted by $s_i$, we compose a system of $n$ spins, denoted by $\sigma_i$, where the states of bits and spins correspond as $0 \to +1$ and $1 \to -1$. If a parity check of the code acts on bits as $\mathbf{h} = s_i \oplus s_j \oplus \dots$, then the corresponding term in the Hamiltonian is $h = \sigma_i \sigma_j \dots$. From the set of parity checks of the code, we construct the Hamiltonian:
\begin{align}\label{eq:hamiltonian_classical}
    H = - \sum_{h_i} \prod_{j \in h_i} \sigma_i.
\end{align}

We take this Hamiltonian and evolve it under Monte-Carlo Metropolis dynamics. As codewords are the states which satisfy all the parity checks, they are the zero energy ground states of the Hamiltonian. 
Robustness to thermal noise and self-correction hence translate to the system remaining in a low energy density configuration for a time which grows with system size at some temperature $\tau$. We consider only Hamiltonians constructed from codes which contain all linearly independent checks. This implies analyticity of the partition function and free energy at all temperatures, allowing the free energy to be calculated directly \cite{mezard_montanari,placke2024topologicalquantumspinglass,placke2025expansioncreatesspinglassorder}. 


In the left panel \cref{fig:selfcorrection}, we show the results of a "heating-cooling" experiment done on diffusion codes with $n = 4888, m= 4000, \wbit = 9, \wcheck = 11$. In the heating part of the experiment, we initialize the code in the all zero code state and then evolve the system under Metropolis dynamics at some temperature $\tau$. This $\tau$ is slowly increased from $0$ in increments of $\Delta \tau = 0.05$ from $\tau = 0$ to $\tau = 4$. Here, one sweep corresponds to an average of $n$ Metropolis updates. Specifically, at each temperature, we do the following: (1) Evolve at temperature $\tau$ for 1000 sweeps. (2) Evolve for some additional time $t_{\rm eq}$ sweeps, taking a sample of the energy every 10 sweeps. \cref{fig:selfcorrection} (left) then shows the energy as a function of temperature recorded during step (2), but averaged over 10 instances of the code and repeated 10 times for each code. For an equilibration time of $t_{\rm eq}$ sweeps, each point in the figure then is the average of $10 \cdot 10 \cdot t_{eq} / 10$ data points. (3) When moving to the next temperature, the final state of the previous temperature is used as the new initial state. As comparison to the measured energy densities, we plot the theoretically predicted equilibrium value as the solid line in \cref{fig:selfcorrection} (left). We observe that at low temperature, the system remains close to 0 code state, indicating a very slow relaxation to equilibrium. Only at some higher temperature $\tau_m$ does the equilibration time exceed the time needed to reach the equilibrium energy density. The transition seems to be abrupt at $\tau_m$, as the recorded energy density of the system jumps to the thermodynamic equilibrium curve. 

The above behavior under ``heating'' is consistent with a memory transition at $\tau_m$. 
In order to confirm this, we also determine the ``memory time'' of the code as a function of system size and temperature.
To this end, we subject a code state to thermal noise at some temperature $\tau$. We hold the system at that temperature and repeatedly check whether the code remains correctable using the $\mathsf{flip}$ decoder. The first time at which the decoder fails to recover the code state is the empirically measured memory time. 
Below the memory transition temperature $\tau_m$, we expect the memory time to grow super-polynomialy with system size, while the system size dependence should be weak (at most logarithmic, at high temperatures). We note that the memory time measured and will depend on the decoder used for readout.

We expect from \cref{theorem:self correction from confinement} that the time to produce an uncorrectable error starting from a code state should be a stretched exponential in system size, given that $\delta(n)$ is a sublinear power law in system size. 
We show the results of a memory time experiment in the right panel of \cref{fig:selfcorrection}. We intialize the system in the all 0 code state and then subject the system to thermal noise. Then, every 10 time steps, we attempt to correct a copy of the system via the $\mathsf{flip}$ decoder. If the $\mathsf{flip}$ decoder converges back to the all 0 state, then the experiment is continued. If the $\mathsf{flip}$ decoder converges to a stopping set or to a different codeword, then the experiment is terminated and the time is recorded. The average of these recorded times across many experiments is what we report as the memory time of the code at a particular temperature. We note that the recorded memory times follow an exponential distribution, similar to that in \cite{Bravyi_2013}. In the panel, we show how this memory time scales with system size at different temperatures. Each point is the average of the 1000 memory time experiments
. We observe that for $\tau < 2$, the memory time grows with system size with a scaling that is consistent with a stretched exponential. In particular, since in these experiments the codes are constructed using a diffusion time of $T = N^1$, the expansion scale we estimate to be $\beta = 1/2$ (i.e. $\delta(n) \sim n^{\beta} = \sqrt{n}$). This corresponds to a memory time $t_{\rm mem} = \exp[\Omega(\sqrt{N})]$. Due to the difficulties of fitting a stretched exponential, we do not attempt an exact fit of the data, but the dashed lines are functions of the form $\sim \mathcal{O}\left( e^{\sqrt{N}} \right)$. 
Our data is indeed consistent with such a scaling of the memory time with system size. 

In the left panel of \cref{fig:selfcorrection}, we also show the results of a "cooling" experiment, in which the system is initialized in a high temperature state at $\tau = 4$ and then slowly annealed down in increments of $\Delta \tau = 0.05$ to $\tau = 0$. As in the heating experiment, at each temperature, we allow the system to equilibrate for 1000 sweeps before sampling and during the sampling period, we hold the system for a time specified by the equilibration time and sample the energy density every 10 sweeps and collect the average. Again, the experiment was performed on 10 different instances codes, and repeated 10 times per code and so each data point in the figure then is the average of $10 \cdot 10 \cdot t_{eq} / 10$ samples. 

Here, we observe that at some temperature $\tau_G$, the system falls off the equilibrium curve and remains in a high energy state. This indicated the onset of glassy behavior in the system. We observe that this behavior persists as we increase the equilibration time per temperature. This is at least consistent with the system undergoing a glass transition at low temperature. The existence of a glass transition in random LDPC codes is well appreciated in the physics literature \cite{mezard_montanari}, and was recently revisited in \cite{placke2025expansioncreatesspinglassorder}. There, it was shown that the glassiness of the codes originates from their expansion property. Our numerical experiments suggest that this behavior is maintained even in for codes with only smaller set expansion. 

\subsubsection{Heating and cooling dynamics of hypergraph products of diffusion codes} \label{sec:quantum self correction}

\begin{figure}
    \centering
    \includegraphics[width=0.49\linewidth]{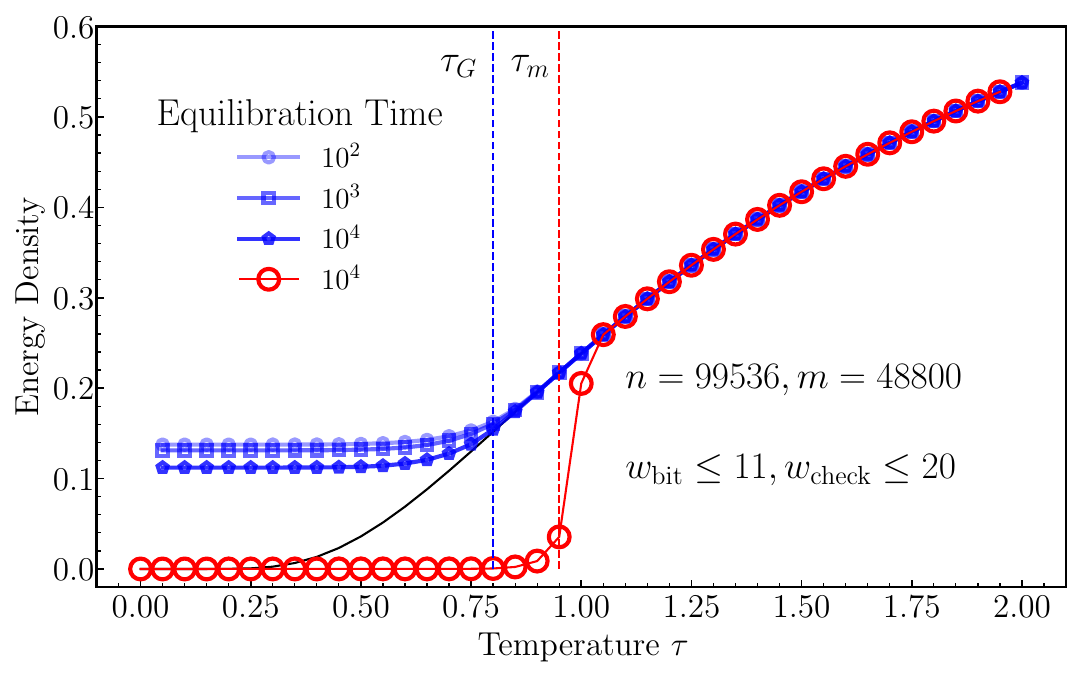}
    \caption{\textbf{Heating-Cooling Experiments for Hypergraph Product Quantum Diffusion Codes}: 
    In this panel we show heating and cooling experiments done on hypergraph products of $\wbit = 9, \wcheck = 11, T = N$ diffusion codes. At $\tau_m > 0.9$, we observe an abrupt transition where the code jumps from near 0 energy density to the predicted thermal equilibrium value. When cooling the system back down to $\tau = 0$, the system becomes trapped in a local minimum and falls off equilibrium.
    }
    \label{fig:heating cooling quantum}
\end{figure}

We now present some numerical results for quantum codes obtained as the hypergraph product of diffusion codes. Doing numerics on these codes is tricky because the taking the hypergraph product of a code with itself squares the number of bits, and hence even the smallest instances with good properties are very large. In particular, as remarked on at the end of \cref{sec:decoding benchmarks}, it is important that, for the classical codes used to take the hypergraph product, $T \gg \wbit, \wcheck$, where here $\wbit$ and $\wcheck$ refer classical code's bit and check degree. If we generate classical diffusion codes that are large enough to perform well, the quantum codes obtained by taking hypergraph products are very large. For a $n = 244, m = 200, \wbit \leq 9, \wcheck \leq 11$ diffusion code, the hypergraph product of this code with itself is a quantum code with $n = 99536$, $m = 48800$, $\wbit \leq 11$, and $\wcheck \leq 20$. It is thus numerically challenging to utilize either the BP decoder or the  $\mathsf{flip}$ decoder. 

Nevertheless, we can at least run a "heating-cooling" experiment as described before, and check whether the results are consistent with self correction and glassiness as observed for the classical codes. The results of such an experiment are shown for $n = 99536, m = 48800, \wbit \leq 11$ and $\wcheck \leq 20$. In both the heating and cooling experiments, at each temperature, we allow the system to equilibrate for 1000 sweeps before sampling. During the sampling period, we hold the system for a time specified by the equilibration time and sample the energy density every 10 sweeps and collect the average.
Here again, one sweep corresponds to an average of $n$ Metropolis updates. We again performed this experiment for 10 different instances, repeating the experiment 10 times per instance and so each point in the figure then is the average of $10 \cdot 10 \cdot t_{eq} / 10$ data points.

In the heating experiment, we again initialize the code in the all zero code state and then evolve the system under Metropolis dynamics at some temperature $\tau$. The $\tau$ is slowly increased from $0$ in increments of $\Delta \tau = 0.05$ up to $\tau = 2$ and at each temperature in our experiment, this system is allowed to equilibrate for a time specified by the equilibration time. As in the classical codes, we observe that at low temperatures $\tau < \tau_m$, the system remains close to 0 code state. Only at higher temperatures $\tau > \tau_m$ does the system reach the equilibrium energy density in the time available during the experiment. This is consistent with a memory transition at $\tau_m$, as observed in the classical codes. 

Finally, we also show the results of a "cooling" experiment in which the system is annealed down to $\tau = 0$ from som high temperature. We again observe that the energy plateaus below some temperature $\tau_G$, and that this effect persists as we increase the equilibration time. This observation is consistent with the prediction that the quantum code
undergoes a transition to a \textit{topological quantum spin glass}, a long-range entangled version of the spin glass phase recently in \cite{placke2024topologicalquantumspinglass}. This phase was reported to occur in quantum expander codes with sufficiently strong expansion. Our work here suggests that this uniquely quantum spin glass phase may also occur in more Euclidean settings when allowing for long-range interactions.

\section{Conclusion and Future Direction}

In this work, we have proved the existence of graphs which expand on a sub-extensive scale, and introduced a family of codes, diffusion codes, providing a tunable trade-off between the expansion scale and the resulting non-locality with respect to some underlying non-euclidean geometry. 
The expansion properties of these codes guarantees a lower bound on the code distance, confinement, and self correction. Furthermore, we showed that quantum codes constructed from hypergraph products of these codes inherit the expansion properties and their consequences. The resulting quantum codes are hence self-correcting, allow efficient decoders against random errors and single-shot error correction. At the same time, they have arbitrarily small power-law non-locality when embedded in an pre-specified euclidean geometry. The code parameters such as distance and memory time under passive decoding scale directly with the degree of expansion and hence non-locality of the code.

We conclude this work by discussing some open questions for future work.

\begin{enumerate}
    \item \textbf{Proof of expansion in diffusion codes on arbitrary graphs}: The diffusion of edges during the construction of diffusion codes may always be reduced to a SEP on some graph. The cycle graph allowed for a convenient organization of states in terms of a gap vector on which we were able to prove a monotonicity condition. This monotonicity condition on the gap vector proved crucial to prove expansion. On arbitrary graphs however, it is not obvious whether our strategy will extend. 

    On $d$-dimensional tori, there may be a way forward. $d$-dimensional tori admit a natural Euclidean metric using which one may organize the state in terms of interparticle distances. The challenge however is in proving a monotonicity condition that allows the reduction of the problem to a smaller system. On the cycle graph, due to the 1D geometry the interparticle distances could be easily organized into the sequential distances in the gap vector. As soon as one moves to 2D however, the particles can move around one another and so one must keep track of all interparticle distances. We therefore believe the 2D problem will be as difficult as the $d$-dimensional problem.

    \item \textbf{Fewer nonlocal checks}: In the diffusion code construction, although non-locality can be turned, almost all edges are of a macroscopic length. This remains true when considering the quantum codes constructed from them. As mentioned earlier, recent results by Dai and Li discuss the degree of nonlocality needed for 2D quantum codes in order to surpass the BPT bounds \cite{BPTbounds1,BPTbounds2,dai2024localityvsquantumcodes}. To reiterate, they prove that any $[n,k,d]$ 2D quantum stabilizer code with $kd^2 \geq \mathcal{O}(n)$ must have at least $c_0 \cdot \max (k,d)$ interactions of length at least $c_0 \cdot \max \left(\frac{d}{\sqrt{n}}, \left(\frac{kd^2}{n} \right)^{1/4} \right)$. The diffusion codes may be tuned to meet the requisite length scale, but a simple extension of \cref{lemma:check distance} will show that almost surely the length of every interaction will be of this length, not just the minimum number required. From this perspective, our quantum codes are sub-optimal. We speculate that their may be other stochastic processes or modifications to the interchange process (such as introducing a weak attractive interaction) which could lead to a more tailored edge-length distribution. In particular, it would be ideal to be able to tune the geometric size distribution of the checks, for example to have the distribution be exponential. Numerics suggest that this would enable fault-tolerant implementation of such codes even using 2D local circuits \cite{Berthusen2024partialsyndrome, PRXQuantum.6.010306}.

    \item \textbf{Higher Dimensional Expansion}: As a result of the discovery of quantum expander codes, there has been a recent explosion of interest in the construction of higher dimensional expanders \cite{ lubotzky2017highdimensionalexpanders}. One may construct these higher dimensional expanders through products of regular expander graphs, or via explicit construction. There are however very few random constructions that construct these higher dimensional expanders from scratch, and the constructions that currently exist typically come with many caveats, such as unbounded degree \cite{liu2022localglobalexpansionrandom}. We speculate whether one can apply similar stochastic process on simplicial complexes in order to obtain the requisite connectivity to generate a higher dimensional expander.

    \item \textbf{Practical Implementations}: Part of the motivation for our work has also been the development of quantum simulator platforms and near term quantum devices \cite{doi:10.1126/science.abi8378, bluvstein2023}. These new experimental platforms have shown rapid development in recent years and have also been used for error correction experiments \cite{googletoriccode2024}. The long range connectivity of quantum expander codes makes it challenging to implement these on current platforms, though there is some intriguing work in this direction \cite{periwalkunkelcooper2024, guo2025selfcorrectingquantumcodesneutral}. Given that the degree of nonlocality may be tuned in the diffusion codes, we are curious whether they may be more easily implemented in current platforms such as neutral atom simulators. The construction of diffusion codes also suggests a natural SWAP circuit for syndrome extraction, though the circuit depth is sub optimal \cite{delfosse2021boundsstabilizermeasurementcircuits}. Speeding up the mixing time of the underlying shuffling process (e.g. via lifting) would however directly translate into shorter syndrome extraction circuits.
\end{enumerate}
\ifanon
\else

\section{Acknowledgments}
The authors thank Daniel Fisher, Louis Golowich, Steven Kivelson, Yaodong Li, Nicholas O'Dea, Akshat Pandey and Charles Stahl for  helpful discussions. We are especially grateful to Nikolas Breuckmann, Tibor Rakovszky and Grace Sommers for past collaboration.  

A.S. acknowledges support from the National Science Foundation Graduate Research Fellowship and the ARCS Foundation for ARCS Scholar funding.
Numerical simulations were performed on Stanford Research Computing Center’s Sherlock cluster.
V.K. acknowledges support from the Packard Foundation through a Packard
Fellowship in Science and Engineering and the US Department of Energy, Office of Science under Award No DE-SC0019380. 
B.P. acknowledges funding through a Leverhulme-Peierls Fellowship at the University of Oxford and the Alexander von Humboldt foundation through a Feodor-Lynen fellowship.

\fi

\section{References}
\bibliographystyle{alpha}
\bibliography{main}

\end{document}